\pdfoutput=1
\documentclass[10pt,a4paper]{article}
\usepackage{amssymb,amsmath,amsthm,mathtools,bm}
\usepackage{graphicx}
\usepackage{a4wide}
\usepackage{float}
\usepackage{tikz}
\usepackage[hidelinks]{hyperref}
\graphicspath{{./figures/}}
\newtheorem{theorem}            {Theorem}
\newtheorem{lemma}     [theorem]{Lemma}
\newtheorem{definition}[theorem]{Definition}

\newtheorem{example} [theorem]          {Example}

\title{Qualitative Stability and Synchronicity Analysis of Power Network Models in Port-Hamiltonian Form}
\author{Volker Mehrmann\footnotemark[1] \, and Riccardo Morandin\footnotemark[2] \, and Simona Olmi\footnotemark[3] \, and\\ Eckehard Sch\"oll\footnotemark[4]}
\begin{document}
\maketitle
\renewcommand{\thefootnote}{\fnsymbol{footnote}}
\footnotetext[1]{
Institut f\"ur Mathematik MA 4-5, TU Berlin, Str.~des 17.~Juni 136, D-10623 Berlin, FRG.
\textit{Email address: }\texttt{mehrmann@math.tu-berlin.de}. Supported by Deutsche Forschungsgemeinschaft via Project A2 within SFB 910 and by {\it Einstein Foundation Berlin} via the Einstein Center ECMath.
}
\footnotetext[2]{
Institut f\"ur Mathematik MA 4-5, TU Berlin, Str.~des 17. Juni 136, D-10623 Berlin, FRG.
\textit{Email address: }\texttt{morandin@math.tu-berlin.de}. Supported by {\it Einstein Foundation Berlin} via the Einstein Center ECMath.
}
\footnotetext[3]{
Institut f\"ur Theoretische Physik, Sekr.~EW 7-1, TU Berlin, Hardenbergstr.~36, D-10623 Berlin, FRG
\& CNR - Consiglio Nazionale delle Ricerche - Istituto dei Sistemi Complessi, 50019, Sesto Fiorentino, Italy.
\textit{Email address: }\texttt{simona.olmi@gmail.com}. Supported by Deutsche Forschungsgemeinschaft via Project A1 within SFB 910.
}
\footnotetext[4]{
Institut f\"ur Theoretische Physik, Sekr.~EW 7-1, TU Berlin, Hardenbergstr.~36, D-10623 Berlin, FRG.
\textit{Email address: }\texttt{schoell@physik.tu-berlin.de}. Supported by Deutsche Forschungsgemeinschaft via Project A1 within SFB 910.
}

\begin{abstract}
In view of highly decentralized and diversified power generation concepts, in particular with renewable energies such as wind and solar power, the analysis and control of the stability and the synchronization of power networks is an important topic that requires different levels of modeling detail for different tasks. A frequently used qualitative approach relies on simplified nonlinear network models like the Kuramoto model.
Although based on basic physical principles, the usual formulation in form of a system of coupled ordinary differential equations is not always adequate. We present a new energy based formulation of the Kuramoto model as port-Hamiltonian system of differential-algebraic equations. This leads to a very robust representation of the system with respect to disturbances, it encodes the underlying physics, such as the dissipation inequality or the deviation from synchronicity, directly in the structure of the equations, it explicitly displays all possible constraints and allows for robust simulation methods. Due to its systematic energy based formulation the model class allows easy extension, when further effects have to be considered, higher fidelity is needed for qualitative analysis, or the system needs to be coupled in a robust way to other networks.
We demonstrate the advantages of the modified modeling approach with analytic results and numerical experiments.
\end{abstract}

\renewcommand*{\thefootnote}{\arabic{footnote}}
\setcounter{footnote}{0}

\section{Introduction}
The increased percentage of renewable energies, such as wind and solar power, and the decentralization of power generation makes the stability and synchronization control of modern power systems
increasingly difficult. To address different control and optimization tasks, there are many different approaches to model power networks; we will briefly present a power grid model hierarchy of
differential-algebraic systems.  At the lowest levels of such a model hierarchy  simplified nonlinear network models like the Kuramoto model are placed which are often used  for a qualitative
analysis of the network behavior \cite{FilNP08,KunBL94}.

The usual formulation of the Kuramoto model in form of a coupled system of ordinary differential equations as in \cite{DorCB13,FilNP08,KunBL94,NisM15,OlmNBT14,OlmT16,RohSTW12,SalMV84} is, however, not always appropriate,
because physical properties like the conservation of energy and momentum, or Kirchhoff's node conditions are only implicity represented in the equations, and thus in numerical simulation or control
approaches they may be violated and lead to un-physical behavior. To prevent this, we present a new energy based formulation of the Kuramoto model as port-Hamiltonian system of differential-algebraic equations.
This model is a very robust representation of the system with respect to disturbances, since it encodes the underlying laws of physics in the algebraic and geometric structure of the equations.
It allows for the development of structure preserving methods that satisfy the physical laws after discretization and in finite precision arithmetic leading to robust simulation and control methods.
The energy based modeling approach allows for easy model refinement, as well as interconnection with other systems from different physical domains. We illustrate the new modeling approach with analytic
results and numerical experiments and  indicate how this approach can be generalized to also allow quantitative analysis.

The basis for our approach are {\em differential-algebraic equations (DAEs)}, also called  {\em descriptor systems} in the control context. They have become a paradigm for the modeling of systems
in different physical domains  and they are incorporated in automated modeling frameworks such as  {\sc modelica}\footnote{\url{https://www.modelica.org/}}.
Descriptor systems allow the explicit representation of constraints and interfaces in the model, see \cite{CamKM11,KunM06}.
In the most general nonlinear setting they have the form
\begin{equation} \label{nldesys}
  F ( t, x,\dot x, u) =0,
\end{equation}
typically together with an initial condition  $x(t_0)=x_0$ and an output equation
\begin{equation}
 y= G(t,x,u).  \label{nldestateout}
\end{equation}
Here, denoting by  $C^0({\mathbb I},{\mathbb R}^{m})$ the set of continuous functions
from a compact time interval $ \mathbb I\subseteq \mathbb R$ to ${\mathbb R}^{m}$, the function  $x$ represents the state, $u$ the input, and $y$ the output of the system. Although more general function spaces can be considered, we assume
that ${F}\in C^0({\mathbb I}\times{\mathbb D}_x\times{\mathbb D}_{\dot x}\times{\mathbb D}_u,{\mathbb R}^\ell)$ is sufficiently smooth, and that
${\mathbb D}_x,{\mathbb D}_{\dot x}\subseteq{\mathbb R}^n$,
${\mathbb D}_{u}\subseteq{\mathbb R}^m$, ${\mathbb D}_y\subseteq{\mathbb R}^p$ are open sets. This most general form of descriptor system is used in the general mathematical analysis and general numerical methods,
see \cite{KunM06}, but it does not display the explicit constraints, e.g.~balance laws, or interface conditions. Furthermore, there may exist hidden constraints or consistency requirements,
which makes further reformulations or regularizations necessary, see \cite{CamKM11,KunM06,LamMT13}.

Another important recent development is the use of \emph{energy based modeling} via bond graphs \cite{Bre08,CouJMTB08}, as implemented in the automated modeling package {\sc 20-sim}\footnote{\url{http://www.20sim.com/}}.
The resulting systems have a \emph{port-Hamiltonian  (pH) structure}, see e.g.~\cite{GolSBM03,JacZ12,OrtSMM01,Sch04,Sch06}, that encodes the underlying physical principles, such as conservation laws, passivity, or stability  directly into the algebraic and geometric structure of the system model.  Ordinary {\em pH systems} have the form
\begin{equation}
\begin{split}
\dot x&=\left(J-R\right)\nabla_{\!x}{\mathcal H}(x)+(B-P)u, \\
y&=\ (B+P)^T \nabla_{\!x}{\mathcal H}(x)  + (S+N) u. \label{PHdef}
\end{split}
\end{equation}
Here the \emph{Hamiltonian function} ${\mathcal H}(x)$  describes  the distribution of internal energy among energy storage elements of the system; $J=-J^T \in \mathbb R^{n,n}$ is the \emph{structure matrix}
describing energy flux among energy storage elements within the system; $R=R^T\in \mathbb R^{n,n}$ is the \emph{dissipation matrix} describing energy dissipation/loss in the system;  $B\pm P\in\mathbb{R}^{n,m}$ are \emph{port} matrices,  describing the manner in which energy enters and exits the system, and $S+N$, with $S=S^T\in \mathbb R^{m,m}$ and $N=-N^T \in \mathbb R^{m,m}$, describes the direct \emph{feed-through} from input to output.
All coefficients $J,R,B,S,N$ can depend on the state $x$ and also explicitly on the time $t$ and also can be infinite dimensional operators. Furthermore, for pH systems it is required that
\begin{equation} \label{Kdef}
W=\left[\begin{array}{lc}
R & P \\[1mm]
P^T & S
\end{array}\right] \geq 0,
\end{equation}
where we write $W>0$ (or $W\geq 0$) to denote that a real symmetric matrix $W$ is positive definite (or positive semi-definite). In contrast to \emph{Hamiltonian systems}, the \emph{conservation of energy} for Hamiltonian systems is replaced by the {\em dissipation inequality}
\begin{equation}  \label{DissipIneq}
 {\mathcal H}(x(t_1))-{\mathcal H}(x(t_0)) \leq \int_{t_{0}}^{t_{1}} y(t)^Tu(t)\ dt,
\end{equation}
which shows that \eqref{PHdef} is a {\em passive} system, see \cite{ByrIW91} and, since ${\mathcal H}(x)$ defines a Lyapunov function, minimal (in the sense of system theory) pH systems are implicitly Lyapunov stable \cite{HinP05,Wil72a,Wil72b}. A major advantage of pH systems in the context of power system modeling is that pH systems are closed under \emph{power-conserving interconnection}, which allows to build-up models in a modularized way, see \cite{CerSB07}, and  Galerkin projection  \cite{BeaG11,GugPBS12,PolS10}, which allows systematic discretization and model reduction.

To include interface conditions or node conditions like Kirchhoff's laws in a pH system, in \cite{BeaMXZ17,Sch13} ordinary pH systems have been extended to {\em port-Hamiltonian differential-algebraic equations (descriptor systems)} (pHDAEs), leading to the following definition, which we present here in the general linear time-varying form,
where we denote by $C^j({\mathbb I},{\mathbb R}^{m})$ the set of $j$-times continuously differentiable functions
from a compact interval $\mathbb I\subseteq \mathbb R$ to ${\mathbb R}^{m}$.
\begin{definition}\label{def:pHDAE}
A linear variable coefficient descriptor system of the form
\begin{equation}
\begin{split}
E \dot x &=\left [(J-R) Q -E T\right ]x + (B-P)u, \\
y&= (B+P)^T Q x  + (S+N) u, \label{pHDAE}
\end{split}
\end{equation}
with $E,Q\in C^1({\mathbb I},{\mathbb R}^{n,n})$, $J,R, T \in C^0({\mathbb I},{\mathbb R}^{n,n})$,
$B,P\in C^0({\mathbb I},{\mathbb R}^{n,m})$, $S,N\in C^0({\mathbb I},{\mathbb R}^{m,m})$ and $S=S^T$, $N=-N^T$
is called \emph{port-Hamiltonian differential-algebraic system (pHDAE)} if the following properties are satisfied:
\begin{itemize}
\item [i)]  For all $t\in \mathbb I$, $Q^T(t)E(t)=E^T(t)Q(t)\in C^{1}(\mathbb I, \mathbb R^{n,n})$ and
\begin{equation*}
 \frac{d}{dt}( Q^T(t) E(t)) = Q^T(t)[E(t) T(t) - J(t)Q(t)] + [E(t) T(t) - J(t)Q(t)]^TQ(t);
\end{equation*}
 \item[ii)] the \emph{Hamiltonian} function $
\mathcal H(x): = \frac 12 x^TQ^TEx:  C^1({\mathbb I},{\mathbb R}^{n})\to C^1({\mathbb I},\mathbb R)$
satisfies ${\mathcal H}(x(t))\geq h_0\in \mathbb R$ uniformly for all $t\in \mathbb I$ and all solutions $x$ of \eqref{pHDAE};
\item [iii)] for all $t\in \mathbb I$, $W=W^T\geq0$, where
\begin{equation} \label{Wdef}
W:=\left[\begin{array}{lc}
Q^T R Q& Q^T P \\[1mm]
P^T Q & S
\end{array}\right]\in C^0({\mathbb I},\mathbb{R}^{n+m,n+m}).
\end{equation}
\end{itemize}
\end{definition}
For a general nonlinear pHDAE of the form \eqref{nldesys}, and a Hamiltonian function
$\mathcal H(x)$, one requires that Definition~\ref{def:pHDAE} holds locally, i.e.~for a given input $u(t)$ and associated trajectory $x(t)$, the Hessian $\mathcal H_{xx}(x)$ can be expressed locally as $E^TQ$, where
$E= F_{\dot x}(t)$,  $F_x(t)=(J-R)Q -E T$, $F_u(t)=B-P$,
$G_x(t)= (B+P)^TQ$, $G_u(t)=S+N$, with $E,J,Q,R, T \in C^0({\mathbb I},{\mathbb R}^{n,n})$, $B,P\in C^0({\mathbb I},{\mathbb R}^{n,m})$,  $S=S^T$, $N=-N^T\in C^0({\mathbb I},{\mathbb R}^{m,m})$.
It has been shown in \cite{BeaMXZ17} that pHDAEs are invariant under time-varying equivalence transformations, they again satisfy the dissipation inequality \eqref{DissipIneq},
and they allow for structure preserving regularization and reformulation as  it was suggested for general DAEs in \cite{CamKM11,KunM06}.
We will not discuss these general results here, but address them in the specific context of power network models where the equations simplify significantly.

\section{A model hierarchy for power networks}\label{sec:Modelhiearchy}
In this section we briefly discuss a model hierarchy of several power network models, which we then turn into a pHDAE formulation in the following section.

Consider a power network of $n$ generators and loads, both represented by oscillators, connected through transmission lines.
If $A=[a_{jk}]\in\mathbb R^{n\times n}$ is the adjacency matrix, then the network can be described by a system of the form
\begin{equation}
\begin{alignedat}{3}
  m_j\dot\theta_j\ddot\theta_j &= - d_j\dot\theta_j^2 - v_jI_j + P_j, &\qquad& \text{for }1\leq j\leq k,\\
  L_{jk}\dot \imath_{jk} &= -R_{jk}\imath_{jk} + v_j - v_k, &\qquad& \text{for }1\leq j<k\leq n, \quad a_{jk}\neq 0,\\
  0 &= \textstyle-I_j + \sum_{k\neq j}a_{jk}\imath_{jk}, &\qquad&\text{for }1\leq j\leq n.\label{inst_power}
\end{alignedat}
\end{equation}
where $\theta_j$ is the \emph{phase angle} and $v_j=V_j\cos\theta_j$  the \emph{voltage} of the $j$-th oscillator, with $V_j>0$ being the \emph{voltage magnitude};  $I_j$ is the \emph{current} passing through the $j$-th oscillator entering the circuit;
$\imath_{jk}$ is the \emph{current through the transmission line} connecting oscillators $j$ and $k$;
$P_j\in\mathbb R$ is the \emph{exchange of power} of the $j$-th oscillator with the environment ($P_j>0$ for generators, $P_j<0$ for loads);
$m_j>0$ and $d_j>0$ are the \emph{angular mass} and the \emph{damping constant} of the $j$-th oscillator, respectively;  $L_{jk}>0$ and $R_{jk}>0$ are the \emph{inductance} and \emph{resistance} of the transmission line connecting oscillators $j$ and $k$, respectively. Moreover, we assume that $m_j$, $d_j$, $L_{jk}$ and $R_{jk}$ are constant in time. This model is called the \emph{instantaneous power model}, since the first equation of \eqref{inst_power} represents the power balance of each node at every instant of time.

In many real power network applications, one expects all angular velocities $\dot\theta_j$ to be close to a constant  frequency $\Omega$ most of the time. 
When this is satisfied, and $V_j,P_j$ do not vary too much, the system will be very close to steady state.
If we assume steady state, and $V_j,P_j$ to be constant in time, then
the second equation in \eqref{inst_power} can be solved explicitly,
and the electrical power $v_iI_j$ in the first equation of \eqref{inst_power} can be replaced by the real power.
This leads to the so-called \emph{real power model}
\begin{equation}\label{real_power}
  m_j\dot\theta_j\ddot\theta_j = - d_j\dot\theta_j^2 - \sum_{k=1}^nr_{jk}\cos(\theta_k-\theta_j)
  + \sum_{k=1}^ng_{jk}\sin(\theta_k-\theta_j) + P_j,
\end{equation}
for $j=1,\ldots,n$, where $R=R^T=[r_{jk}]\geq 0$, $G=G^T=[g_{jk}]\leq 0$, with coefficients $r_{jk}$ and $g_{jk}$ depending on $R_{jk},V_j,V_k,\Omega$ and $L_{jk},V_j,V_k,\Omega$ respectively, for $j,k=1,\ldots,n$.
Note that in this model the current $\imath$ is included only implicitly.

Since the entries $r_{jk}$ of the resistance matrix are usually negligible compared to the other coefficients, one could assumo $r_{jk}\equiv0$.
Furthermore, for small time intervals one may approximate $m_j\dot\theta_j\ddot\theta_j\approx m_j\Omega\ddot\theta_j$ and $d_j\dot\theta_j^2\approx d_j\Omega^2+2d_j\Omega\dot\theta_j$ (see \cite{FilNP08}). This leads (up to rescaling), to the system
\begin{equation}\label{eq:genkura}
  m_j\ddot\theta_j = -\tilde d_j\dot\theta_j + \sum_{k=1}^n \tilde g_{jk}\sin(\theta_k-\theta_j) + \tilde P_j,
\end{equation}
for $j=1,\ldots,n$, which we call the \emph{generalized Kuramoto model}. It is a generalization of the standard \emph{Kuramoto model} which consists of
a system of $n$ fully-coupled oscillators satisfying the equations
\begin{equation}\label{eq:stdkura0}
  m_j\ddot\theta_j + d_j\dot\theta_j = \Omega_j + K\sum_{k=1}^n\sin(\theta_k-\theta_j),
\end{equation}
for $j=1,\ldots,n$, where $\theta_j$ denote the phase angles, $m_j>0$ the masses, $d_j>0$ the damping constants, $\Omega_j$ the natural frequencies and $K>0$ the coupling constant of the system.
It should be noted that in \cite{FilNP08} the perturbations $\tilde\delta_j=\theta_j-\Omega t$ are used as variables in \eqref{eq:stdkura0}, instead of $\theta_j$. This can always be done for \eqref{eq:genkura}, up to changing $\tilde P_j$ by a constant, but not for \eqref{inst_power} or \eqref{real_power}.

We summarize all the mentioned models in a \emph{model hierarchy}, see Figure~\ref{fig:hierarchy},
\begin{figure}[htb]
  \centering
  \includegraphics[width=0.35\textwidth]{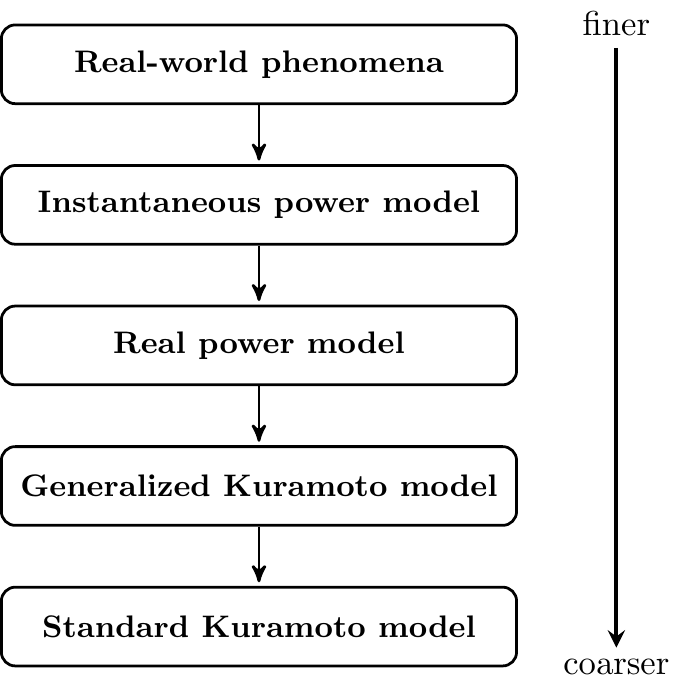}
  \caption{Model hierarchy for power networks}\label{fig:hierarchy}
\end{figure}
where  lower level models result from simplifications of the higher levels.
Using such a model hierarchy allows to adapt the model depending on the task, the accuracy requirements, or the allowed computation time.
For example, to check the stability properties of a synchronous state, the Kuramoto model may be used to achieve high computational efficiency, while the instantaneous power model may be required when accurate quantitative solutions are necessary.
Note that the model hierarchy is by no means complete and it should be extended, when further components (like e.g.~transformer stations) need to be included, or when the parameters of the system (like e.g.~the load) need to be modeled in a stochastic way.

\subsection{The order parameter}\label{sec:orderpara}
The standard Kuramoto model is frequently used  to study qualitatively the synchronization of power  networks
\cite{FilNP08,OlmNBT14,OlmT16}. For this one introduces the complex \emph{order parameter}
\[
  re^{i\phi} = \frac 1n\sum_{j=1}^ne^{i\theta_j}\in\{z\in\mathbb C:|z|\leq1\}.
\]
The value of $r\in[0,1]$ is $1$ when the system is in a fully synchronized state and it is $0$ when it is completely desynchronized. Note that
\[
  r^2 = re^{i\phi}\overline{re^{i\phi}} = \frac1{n^2}\sum_{j,k=1}^ne^{i(\theta_k-\theta_j)} = \frac1{n^2}\sum_{j,k=1}^n\cos(\theta_k-\theta_j),
\]
so introducing new variables $\omega_j\coloneqq\dot\theta_j$, $\rho_j\coloneqq\cos\theta_j$, $\sigma_j\coloneqq\sin\theta_j$, $j=1,\ldots,n$, we obtain
\[
  \sin(\theta_k-\theta_j) = \sigma_k\rho_j-\sigma_j\rho_k, \qquad
  \cos(\theta_k-\theta_j) = \rho_j\rho_j + \sigma_k\rho_k,
\]
and system \eqref{eq:stdkura0} takes the form
\begin{equation}
\begin{split}
  \dot\theta_j &= \omega_j, \\
  m_j\dot\omega_j &= \Omega_j - d_j\omega_j + K\rho_j\sum_{k=1}^n\sigma_k - K\sigma_j\sum_{k=1}^n\rho_k.\label{kuramoto_new}
\end{split}
\end{equation}
In this paper we focus on the bottom two levels of the model hierarchy and show that the generalized Kuramoto model has several advantages compared to the standard Kuramoto model.
However, to improve the robustness of the representations further, in the next section we will generate an energy based formulation of these two models as port-Hamiltonian differential-algebraic systems (pHDAE).

\section{PHDAE formulation of the Kuramoto models}\label{sec:daeph_kuramoto}
In this section we will reformulate the two Kuramoto models \eqref{eq:stdkura0} and \eqref{eq:genkura} as pHDAEs.
Introducing the vectors and matrices $\theta=[\theta_j]$, $\omega=[\omega_j]$, $\rho=[\rho_j]$, $\sigma=[\sigma_j]$, $\Omega=[\Omega_j]$, $M=\operatorname{diag}(m_j)$, $D=\operatorname{diag}(d_j)$, $D_\rho=\operatorname{diag}(\rho_j)$, $D_\sigma=\operatorname{diag}(\sigma_j)$ and $G=K e e^T$, 
where $e\in\mathbb R^n$ is the vector of all ones, we can write \eqref{kuramoto_new} as
\begin{equation}
\begin{split}
  \dot\theta &= \omega, \\
  M\dot\omega &= \Omega - D\omega - D_\sigma G\rho + D_\rho G\sigma. \label{kuramot_vec}
\end{split}
\end{equation}
Introducing the derivatives of $\rho$ and $\sigma$, we get
$\dot\rho = -D_\sigma\omega$, and $\dot\sigma = D_\rho\omega$.
Since computing $\rho$ and $\sigma$ is equivalent to determining $\theta$, we can discard the first equation in~\eqref{kuramot_vec} and obtain a simplified pHDAE formulation $E\dot x = (J-R)Qx + Bu$ of the standard Kuramoto model as
\begin{equation}\label{eq:stdkura1}
  \begin{split}
    M\dot\omega &= \Omega - D\omega - D_\sigma G\rho + D_\rho G\sigma, \\
    \dot\rho &= -D_\sigma\omega, \\
    \dot\sigma &= D_\rho\omega,
  \end{split}
\end{equation}
where $x=[\omega^T,\rho^T,\sigma^T]^T\in\mathbb R^{3n}$, $u=\Omega\in\mathbb R^n$, $E=\operatorname{diag}(M,I,I)=\mathbb R^{3n\times3n}$, $0\leq R=R^T=\operatorname{diag}(D,0,0)\in\mathbb R^{3n\times 3n}$, $Q=\operatorname{diag}(I,-G,-G)\in\mathbb R^{3n\times3n}$, and
\[
  J = -J^T=
  \begin{bmatrix}
    0 & D_\sigma & -D_\rho \\
    -D_\sigma & 0 & 0 \\
    D_\rho & 0 & 0
  \end{bmatrix}
  ,\qquad B =
  \begin{bmatrix}
    I \\ 0 \\ 0
  \end{bmatrix}.
\]
In contrast to the general formulation of pHDAEs in~\eqref{pHDAE}, we have that $E$ and $Q$ are constant in time, and $T=0$. The  \emph{Hamiltonian function} is
\begin{equation}\label{eq:stdkura:ham}
  \mathcal H(x) = \frac12x^TQ^TEx = \frac12\omega^TM\omega - \frac12\rho^TG\rho - \frac12\sigma^TG\sigma,
\end{equation}
and, in particular, we have
\begin{equation*}
  \frac12\rho^TG\rho + \frac12\sigma^TG\sigma
  = \frac12K\sum_{j,k=1}^n(\rho_j\rho_k+\sigma_j\sigma_k)
  = \frac12K\sum_{j,k=1}^n\cos(\theta_k-\theta_j)
  = \frac12Kn^2r^2,
\end{equation*}
so that $\mathcal H(x) = \frac12\omega^TM\omega - \frac12Kn^2r^2$.

Since we have omitted the dependence of $\rho$ and $\sigma$ on $\theta$, the components $\rho$ and $\sigma$ of the solution $x$ implicitly need to satisfy the property
\begin{equation}\label{eq:algcond}
  \rho_j^2(t)+\sigma_j^2(t)=1, \qquad\text{for }j=1,\ldots,n
\end{equation}
for all $t\in \mathbb I$, since $\rho_j=\cos(\theta_j)$ and $\sigma_j=\sin(\theta_j)$ for some functions $\theta_1,\ldots,\theta_n$.
Unfortunately, when applying numerical integrators to \eqref{eq:stdkura1}, implicit relations like \eqref{eq:algcond} are typically not preserved, due to discretization and  roundoff errors, see \cite{KunM07} for a detailed discussion and stabilization techniques to avoid this effect.
One way out of this problem is to add the conditions \eqref{eq:algcond} explicitly to the system, as $n$ new algebraic equations, making the system overdetermined. 
To do so, we introduce a vector of Lagrange multipliers $\mu=[\mu_j]\in\mathbb R^n$. System \eqref{eq:stdkura1} is then equivalent to
\begin{equation}
\begin{split}
  M\dot\omega &= \Omega - D\omega - D_\sigma G\rho + D_\rho G\sigma, \\
  \dot\rho &= -D_\sigma\omega, \\
  \dot\sigma &= D_\rho\omega, \\
  0 &= D_\rho\rho+D_\sigma\sigma-\mu,\label{eq:stdkura2}
\end{split}
\end{equation}
together with initial conditions $x(t_0)=x_0$, satisfying \eqref{eq:algcond} and $\mu(t_0)=e$.
Indeed, by differentiating the 4\textsuperscript{th} equation of \eqref{eq:stdkura2} with respect to $t$, and substituting the 2\textsuperscript{nd} and 3\textsuperscript{rd} equation, we get
\[
  \dot\mu = D_\rho\dot\rho+D_\sigma\dot\sigma = -D_\rho D_\sigma\omega+D_\sigma D_\rho\omega = 0,
\]
so $\tilde x=[\omega^T,\rho^T,\sigma^T,\mu^T]^T$ is a solution of \eqref{eq:stdkura2} if and only if $x=[\omega^T,\rho^T,\sigma^T]^T$ is a solution of $\eqref{eq:stdkura1}$ and $\mu$ is constant.
System \eqref{eq:stdkura2} can again be written as pHDAE $ \tilde E\dot{\tilde x} = (\tilde J-\tilde R)\tilde Q\tilde x + \tilde Bu$, with
\begin{equation*}
  \tilde E =
  \begin{bmatrix}
    E & 0 \\ 0 & 0
  \end{bmatrix}, \
  \tilde J =
  \begin{bmatrix}
    J & 0 \\ 0 & 0
  \end{bmatrix}, \
  \tilde x =
  \begin{bmatrix}
    x \\ \mu
  \end{bmatrix}, \ 
  \tilde Q =
  \left[\begin{array}{ccc|c}
    & & & 0 \\
    & Q & & 0 \\
    & & & 0 \\ \hline
    0 & D_\rho & D_\sigma & -I
  \end{array}\right], \
  \tilde R =
  \begin{bmatrix}
    R & 0 \\ 0 & I
  \end{bmatrix},\ \tilde B =
  \begin{bmatrix}
    B  \\ 0
  \end{bmatrix}.
\end{equation*}
Note that again $\tilde J=-\tilde J^T$, and $\tilde Q^T\tilde E = \tilde E^T \tilde Q$ do not depend on time, and the Hamiltonian has not changed. Furthermore,
\[
  \tilde Q^T\tilde R\tilde Q = 
  \left[\begin{array}{c|c}
    D  & 0 \\ \hline 0 &
    \begin{bmatrix}
      D_\rho \\ D_\sigma \\ -I
    \end{bmatrix}
    \begin{bmatrix}
      D_\rho \\ D_\sigma \\ -I
    \end{bmatrix}^T
  \end{array}\right]\geq 0.
\]

In an analogous way we can formulate the generalized Kuramoto as a pHDAE. We take a slightly more general approach than in  Section~\ref{sec:Modelhiearchy} and start from  a system of $n$ oscillators satisfying the equations
\begin{equation}\label{eq:kuramoto0}
  m_j\ddot\theta_j + d_j\dot\theta_j = \Omega_j + \sum_{k=1}^ng_{jk}\sin(\theta_k-\theta_j), 
\end{equation}
for $j=1,\ldots,n$, where $g_{jk}$  are the entries of $G=G^T\in\mathbb R^{n\times n}$, satisfying  $g_{jk}=g_{kj}\geq0$ for $j\neq k$. These entries represent the strength of the link between
the oscillators $j$ and $k$, and are $0$ when no link is present.
The diagonal entries of $G$ can be chosen freely. We will also denote by $G_0$ the matrix without its diagonal entries, i.e.~$G_0\coloneqq G-\operatorname{diag}(G)$.
Note that if we choose $G=Ke e^T$, then we are in the same situation as in Section~\ref{sec:Modelhiearchy}.

The pHDAE formulation of the generalized Kuramoto system \eqref{eq:kuramoto0} is achieved analogously to that of the standard Kuramoto model, with the Hamiltonian given by \eqref{eq:stdkura:ham},
where the choice of $\operatorname{diag}(G)$ modifies the Hamiltonian by an additive constant, since
\begin{equation*}
  -\frac12\rho^TG\rho-\frac12\sigma^TG\sigma
  = -\frac12\rho^TG_0\rho-\frac12\sigma^TG_0\sigma-\frac12\sum_{j=1}^ng_{jj}(\rho_j^2+\sigma_j^2),
\end{equation*}
and thus
\begin{equation*}
  -\frac12\rho^TG\rho-\frac12\sigma^TG\sigma
  = -\frac12\sum_{j,k}g_{jk}(\rho_j\rho_k+\sigma_j\sigma_k)
  = -\frac12\sum_{j,k}g_{jk}\cos(\theta_k-\theta_j)
\end{equation*}
takes the role of $-\frac12n^2r^2$ in the Hamiltonian.

\subsection{The generalized order parameter}
The modification of the Hamiltonian by an additive constant in the pHDAE formulation of the generalized Kuramoto model suggests to also define a \emph{generalized order parameter} of the form
\[
  \xi \coloneqq \frac1c\sum_{j,k}g_{jk}\cos(\theta_k-\theta_j) + b,
\]
with real constants $b,c>0$, possibly depending on $G$, which take the role of $r^2$ in the standard Kuramoto model. Using the generalized order parameter, the Hamiltonian can be written as
\[
  \mathcal H = \frac12\omega^TM\omega + \frac12c(b-\xi),
\]
where $b$, $c$ need to be chosen to retain consistency with the standard Kuramoto model.

Note that we still have freedom in choosing the diagonal entries of the matrix $G$, that have no influence in the differential equation, and that they change $\xi$ and the Hamiltonian $\mathcal H$ in the same way as $b$.
\begin{lemma}\label{lem:maxval}
Consider the pHDAE formulation of \eqref{eq:kuramoto0} with generalized order parameter
$\xi$. Then the maximal value of $\xi(\theta)$ is given by $c^{-1}\sum_{j,k}g_{jk}+b$, and the maximum is achieved if and only if each connected component of the network, whose topology is defined using $G$ as (weighted) adjacency matrix, is fully synchronized.
\end{lemma}
\begin{proof}
Since $c>0$, $g_{jk}\geq 0$ for $j\neq k$, and $\cos(\theta_k-\theta_j)\leq1$, it is clear that for all $\theta\in\mathbb R^n$ we have
\[
    \xi(\theta)\leq c^{-1}\sum_{j,k}g_{jk}+b.
\]
Equality holds if and only if $\theta_j\equiv\theta_k$ (mod $2\pi$) for all $j,k$ such that $g_{jk}\neq 0$. This condition is equivalent to the following statement. If there is a path through the network defined by $G$ that connects $j$ and $k$, then $\theta_j\equiv\theta_k$ (mod $2\pi$). In turn, this property is equivalent to the full synchronicity of all connected components.
\end{proof}
If we want to preserve the property that $\max_\theta\xi(\theta)=1$, then we must require that $c(1-b)=\sum_{j,k}g_{jk}$, or equivalently $c(1-b)-\sum_jg_{jj}=e^TG_0 e$.
We can also force the Hamiltonian $\mathcal H$ to be non-negative, with minimum value $0$, by choosing $b=1$, i.e.~$\sum_jg_{jj}=-e^TG_0e$.
\begin{lemma}\label{lem:spec}
Suppose that $b=1$ and $g_{jj}=-\sum_{k\neq j}g_{jk}$ for $j=1,\ldots,n$. Then $\max_\theta\xi(\theta)=1$, $\min_x\mathcal H(x)=0$. Furthermore, $G$ is singular and negative semi-definite.
\end{lemma}
\begin{proof}
The particular choice of $b$ and $g_{jj}$ gives the first two assertions.
  We also have $Ge=0$, so $G$ is singular and finally, $G \leq 0$ is an immediate consequence of the Gersgorin circle theorem~\cite{GolV96}.
\end{proof}

To illustrate the previous analysis we present some examples.
\begin{example}\label{ex:ex1}\textup{
  Consider the standard Kuramoto model with $G_0=K(e e^T-I)$. Choose $b=1$ and $[g_{jj}]=-G_0e=-K(n-1)e$, i.e.~$G=K(e e^T-nI)$. To obtain $c>0$ such that $\xi=r^2$, we note that
\[
    \xi = \frac1c\left(K\sum_{j,k}\cos(\theta_k-\theta_j)-Kn^2\right)+1  \\
    = \frac1c(Kn^2r^2-Kn^2)+1 = \frac{Kn^2}c(r^2-1)+1,
\]
so  $c=Kn^2$ gives $\xi=r^2$ and $\mathcal H = \frac12\omega^TM\omega + \frac12Kn^2(1-r^2)$.
  }
\end{example}
In the following, we always assume that $b=1$ and $g_{jj}=-\sum_{k\neq j}g_{jk}$ for $j=1,\ldots,n$. It would be interesting to choose $c(G_0)$ in such a way that $\min_\theta\xi(\theta)=0$ for each choice of $G_0$, but it is not clear how this can be achieved, since
\[
  \min\left({\sum_{j,k}g_{jk}\cos(\theta_k-\theta_j)}\right)
\]
depends strongly on the topology and on the weights of the network.

\begin{example}\textup{
Consider a network with  $3$ oscillators
  \[
    G_0 =
    \begin{bmatrix}
      0 & 1 & 1 \\
      1 & 0 & 0 \\
      1 & 0 & 0
    \end{bmatrix},
    \quad 
    G =
    \begin{bmatrix}
      -2 & 1 & 1 \\
      1 & -1 & 0 \\
      1 & 0 & -1
    \end{bmatrix}.
  \]
While in the standard Kuramoto model the minimum of $\xi$ was achieved for $e^{i\theta_j}$ evenly distributed on the unit circle for $j=1,2,3$, now it is achieved for $e^{i\theta_2}=e^{i\theta_3}=-e^{i\theta_1}$ instead. Indeed, this condition gives that for all $j,k$ such that $j\neq k$ and $g_{jk}\neq0$, we have $\cos(\theta_k-\theta_j)=-1$, a configuration that was not possible in the fully coupled case.
In this case we have
\[
    \sum_{j,k}g_{jk}\cos(\theta_k-\theta_j) = -\sum_{j,k}|g_{jk}| = -8,
\]
while choosing the $e^{i\theta_j}$ evenly distributed would give only
\[
    \sum_{j,k}g_{jk}\cos(\theta_k-\theta_j) = -6.
\]
A similar effect would also happen if we choose
  \[
    G_0 =
    \begin{bmatrix}
      0 & 1 & 1 \\
      1 & 0 & \varepsilon \\
      1 & \varepsilon & 0
    \end{bmatrix},
  \]
  with $\varepsilon>0$ very small.
  }
\end{example}

Since we have
\[
  \sum_{j,k}g_{jk}\cos(\theta_k-\theta_j) \geq -\sum_{j,k}|g_{jk}| = -2e^TG_0e,
\]
and equality can happen for some $G_0$ and $\theta$, one may be tempted to choose $c=2 e^TG_0e$, to ensure $\xi\geq0$.
Unfortunately, that would not be consistent with the standard Kuramoto model, since in that case $e^TG_0e=Kn(n-1)$ but $c=Kn^2$ (see Example~\ref{ex:ex1}).
That suggests to define
\[
  c(G_0) = \frac{n}{n-1}e^TG_0e = \left({1+\frac{1}{n-1}}\right)e^TG_0e,
\]
to guarantee the consistency with the standard Kuramoto model. With this choice, we obtain
\begin{equation*}
  \xi(\theta)
  = 1 - \frac{n-1}{n}\frac{\sum_{j,k}g_{jk}\cos(\theta_k-\theta_j)}{e^TG_0e}
  \geq -1 + \frac2n,
\end{equation*}
so we can guarantee that $\xi\in[-1+\frac2n,1]$, in particular $|\xi|\leq1$.

\subsection{Relative phase angle formulation}
To analyze the dynamics of the system, it is sufficient to consider the relative phase angles $\theta_j-\theta_k$ for $j,k=1,\ldots,n$, instead of the absolute phase angles $\theta_j$.  This has the advantage that it removes some redundancy from the system.
For example, states $(\theta,\omega)$ and $(\theta+\gamma e,\omega)$ with $\theta,\omega\in\mathbb R^n$ and $\gamma\in\mathbb R$ behave in the same way in the absolute system, but are actually the same in the relative system.
To simplify the presentation, for the remainder of this subsection, we suppose to have $n+1$ oscillators $\theta_0,\ldots,\theta_n$ instead of $n$. Then considering all phase angles relative to $\theta_0$, setting
\[
  \hat\theta_j \coloneqq \theta_j-\theta_0, \quad
  \hat\rho_j \coloneqq \cos(\hat\theta_j), \quad
  \hat\sigma_j \coloneqq \sin(\hat\theta_j),\qquad j=1,\ldots,n,
\]
 and $\hat\theta_0=0$, $\hat\rho_0=1$ and $\hat\sigma_0=0$, we obtain that
\begin{equation}
\begin{alignedat}{2}
  \cos(\theta_k-\theta_j)
  &= \cos(\hat\theta_k-\hat\theta_j)
  &&= \hat\rho_j\hat\rho_k+\hat\sigma_j\hat\sigma_k, \\
  \sin(\theta_k-\theta_j)
  &= \sin(\hat\theta_k-\hat\theta_j)
  &&= \hat\rho_j\hat\sigma_k-\hat\sigma_j\hat\rho_k
  \label{eq:relkura:prop0}
\end{alignedat}
\end{equation}
are satisfied for $j,k=0,\ldots,n$.
Introduce the partitioning
\[
  G =
  \begin{bmatrix}
    g_{00} & g_0^T \\
    g_0 & \hat G
  \end{bmatrix},\quad
  g_{00}\in\mathbb R,\quad
  g_0\in\mathbb R^n,\quad
  \hat G\in\mathbb R^{n\times n},
\]
and choose the diagonal as in the previous subsection, so that $Ge=0$, and then $g_0=-\hat G e$ and $g_{00}=-e^Tg_0=e^T\hat Ge$.
Let us partition $M=\operatorname{diag}(m_0,\hat M)$, $\omega=[\omega_0,\hat\omega^T]^T$, $D=\operatorname{diag}(d_0,\hat D)$, $\Omega=[\Omega_0,\hat\Omega^T]^T$ and define $\hat\rho=[\hat\rho_j]_{j=1,\ldots,n}\in\mathbb R^n$, $\hat\sigma=[\hat\sigma_j]_{j=1,\ldots,n}\in\mathbb R$, $D_{\hat\rho}=\operatorname{diag}(\hat\rho)$ and $D_{\hat\sigma}=\operatorname{diag}(\hat\sigma)$. Then the system can be written equivalently as
\begin{align*}
  m_0\dot\omega_0 &= -d_0\omega_0 - e^T\hat G\hat\sigma + \Omega_0, \\
  \hat M\dot{\hat \omega} &= -\hat D\hat \omega
    + D_{\hat\rho}\hat G\hat\sigma
    + D_{\hat\sigma}\hat G(e-\hat\rho) + \hat \Omega, \\
  \dot{\hat\rho} &= -D_{\hat\sigma}{\hat \omega}+\hat\sigma\omega_0, \\
  \dot{\hat\sigma} &= D_{\hat\rho}{\hat \omega}-\hat\rho\omega_0.
\end{align*}
Introducing the new variable $\tilde\rho_j=e-\hat\rho_j$ instead of $\hat\rho_j$, for $j=1,\ldots,n$, and correspondingly $\tilde\rho\in\mathbb R^n$ and $D_{\tilde\rho}\in\mathbb R^{n\times n}$, we obtain the new system
\begin{align*}
  m_0\dot\omega_0 &= -d_0\omega_0 - \hat\sigma^T\hat G\tilde\rho - \hat\rho^T\hat G\hat\sigma + \Omega_0, \\
  \hat M\dot{\hat \omega} &= -\hat D{\hat \omega}
    + D_{\hat\rho}\hat G\hat\sigma
    + D_{\hat\sigma}\hat G\tilde\rho + \hat\Omega, \\
  \dot{\tilde\rho} &= D_{\hat\sigma}{\hat \omega}-\hat\sigma\omega_0, \\
  \dot{\hat\sigma} &= D_{\hat\rho}{\hat \omega}-\hat\rho\omega_0,
\end{align*}
that is again  a pHDAE system of the form $
  \hat E\dot{\hat x} = (\hat J-\hat R)\hat Q\hat x + \hat Bu$,
with
$\hat x=[\omega_0,\omega^T,\tilde\rho^T,\hat\sigma^T]^T$, $u=\Omega$, $\hat E=\operatorname{diag}(m_0,\hat M,I,I)$, $\hat Q=\operatorname{diag}(1,I,-\hat G,-\hat G)$, $\hat R=\operatorname{diag}(d_0,\hat D,0,0)$ and
\[
  \hat J =
  \begin{bmatrix}
    0 & 0 & \hat\sigma^T & \hat\rho^T \\
    0 & 0 & -D_{\hat\sigma} & -D_{\hat\rho} \\
    -\hat\sigma & D_{\hat\sigma} & 0 & 0 \\
    -\hat\rho & D_{\hat\rho} & 0 & 0
  \end{bmatrix}, \qquad
  \hat B =
  \begin{bmatrix}
    1 & 0 \\
    0 & I \\
    0 & 0 \\
    0 & 0
  \end{bmatrix},
\]
with Hamiltonian
\begin{equation*}
  \mathcal H(\hat x) = \frac12\hat x^TQ^TE\hat x 
  = \frac12\omega^TM\omega - \frac12\tilde\rho^T\hat G\tilde\rho - \frac12\hat\sigma^T\hat G\hat\sigma.
\end{equation*}
Note that $\hat J=-\hat J^T$, $\hat R=\hat R^T\geq 0$ and $\hat Q^T\hat E$ is symmetric and constant in time.
Furthermore, we still have $\hat G\leq0$ from the Gersgorin circle theorem, so $\hat Q^T\hat E\geq 0$.
It can also be shown that if the network is connected, then $\hat G<0$ and $\hat Q^TE>0$. This is a consequence of a stronger version of the Gersgorin theorem \cite[Theorem 1.12]{Var11}, since connection would imply $g_0\not\equiv0$.
By applying property \eqref{eq:relkura:prop0}, we get
\begin{align*}
  \sum_{j,k=0}^ng_{jk}\cos(\theta_k-\theta_j)
  &= [1\;\hat\rho^T]
    \begin{bmatrix}
      g_{00} & g_0^T \\ g_0 & \hat G
    \end{bmatrix}
    \begin{bmatrix}
      1 \\ \hat\rho
    \end{bmatrix}
    +
    [0\;\hat\sigma^T]
    \begin{bmatrix}
      g_{00} & g_0^T \\ g_0 & \hat G
    \end{bmatrix}
    \begin{bmatrix}
      0 \\ \hat\sigma
    \end{bmatrix}  \\
  &= g_{00}+2g_0^T\hat\rho+\hat\rho^T\hat G\hat\rho + \hat\sigma^T\hat G\hat\sigma  \\
  &= e^T\hat Ge - 2e^T\hat G\hat\rho + \hat\rho^T\hat G\hat\rho + \hat\sigma^T\hat G\hat\sigma  \\
  &= (e-\hat\rho)^T\hat G(e-\hat\rho) + \hat\sigma^T\hat G\hat\sigma  \\
  &= \tilde\rho^T\hat G\tilde\rho + \hat\sigma^T\hat G\hat\sigma,
\end{align*}
so the Hamiltonian has not changed.

Finally, as in the case of the standard Kuramoto model, we incorporate  the algebraic equations
\[
  \hat\rho_j^2 + \hat\sigma_j^2 = 1, \ j=1,\ldots,n
\]
or equivalently
\[
  -(1+\hat\rho_j)\tilde\rho_j + \hat\sigma_j^2 = 0, \ j=1,\ldots,n.
\]
We introduce again Lagrange multipliers $\mu\in\mathbb R^n$ into the system and obtain
%
\[
  -(I+D_{\hat\rho})\tilde\rho + D_{\hat\sigma}\sigma + \mu = 0.
\]
Proceeding as in the reformulation of the standard Kuramoto model, we obtain a modified pHDAE formulation with the same Hamiltonian.
%
%

\section{Numerical results}

In this section we present some numerical simulation results that illustrate the advantages of our new, extended pHDAE formulation.
In particular, we will concentrate on the Italian high-voltage ($380$ kV) power grid (Sardinia excluded),
which is composed of $n=127$ nodes, divided in $34$ sources (hydroelectric and thermal power plants) and $93$ consumers,
connected by $342$ links \cite{ForFSF12}.
This network is characterized by a quite low average connectivity $n_c=2.865$, due to the geographical distributions of the nodes along Italy.
The map of the Italian high-voltage power grid can be seen at the website of the Global Energy Network Institute\footnote{\url{http://www.geni.org}}.
The data that we used have been extracted from the map delivered by the Union for the Coordination of Transport of Electricity (UCTE)\footnote{\url{https://www.entsoe.eu/map/Pages/no-webgl.html}}.
We represent the power network through the generalized Kuramoto model.
For the sake of simplicity, each node is assumed to have the same angular mass and the same damping constant ($m_j=m$ and $d_j=d$ for $j=1,\ldots,n$). Furthermore, each existing transmission line is assumed to have the same coupling coefficient (that is, we have either $g_{jk}=K$ or $g_{jk}=0$ for all $j,k=1,\ldots,n$ with $j\neq k$).
We distinguish the generators from the consumers by the sign of the power (i.e.~the natural frequency) associated to each node: $\Omega_j>0$ for generators and $\Omega_j<0$ for consumers.
When not specified differently, we will choose as standard parameters in our simulations $m=6$ and $d=1$.

\subsection{Perfectly balanced bimodal distribution}
Up to a change of variables, we can consider the generalized Kuramoto model to have reference frequency $0$.
Then, to have a stable fully locked state as a possible solution of \eqref{eq:genkura} (that is, a solution trajectory such that $\omega(t)\equiv0$), it is necessary that $\sum_{j=1}^n\Omega_j=0$.
We assume all generators to have the same power $\Omega_j=+G$, and all consumers to have the same power $\Omega_j=-C$.
In our simulations we have set $C = 1.0$, $G = 2.7353$, in order to achieve $\sum_{j=1}^n\Omega_j=0$.
This setup corresponds to a Kuramoto model with inertia with perfectly balanced bimodal distribution of the natural frequencies.

In order to validate the proposed model, we integrate the set of equations \eqref{eq:stdkura2} with different integration schemes, and compare the results obtained.
To guarantee consistent energy behavior and conservation of the Lagrange multiplier $\mu$, a geometric iterator should be applied. Our methods of choice are the symplectic Euler method and the implicit midpoint rule; for the latter, its effectiveness and computational time depend on the accuracy requested to the nonlinear solver at each integration step.
We will also employ the explicit midpoint method, which is not a geometric integrator, but it allows for faster simulations than the implicit midpoint, with the same order.
Finally, we use as a reference the integration of original equation for the Kuramoto oscillators with inertia \eqref{eq:stdkura0}, with a $4$\textsuperscript{th} order Runge-Kutta scheme.
For all simulations, when not specified differently, we choose integration step $0.002$, transient time $100$ and simulation time $5000$.

\begin{figure}[htb]
  \centering
  \includegraphics[width=0.5\textwidth]{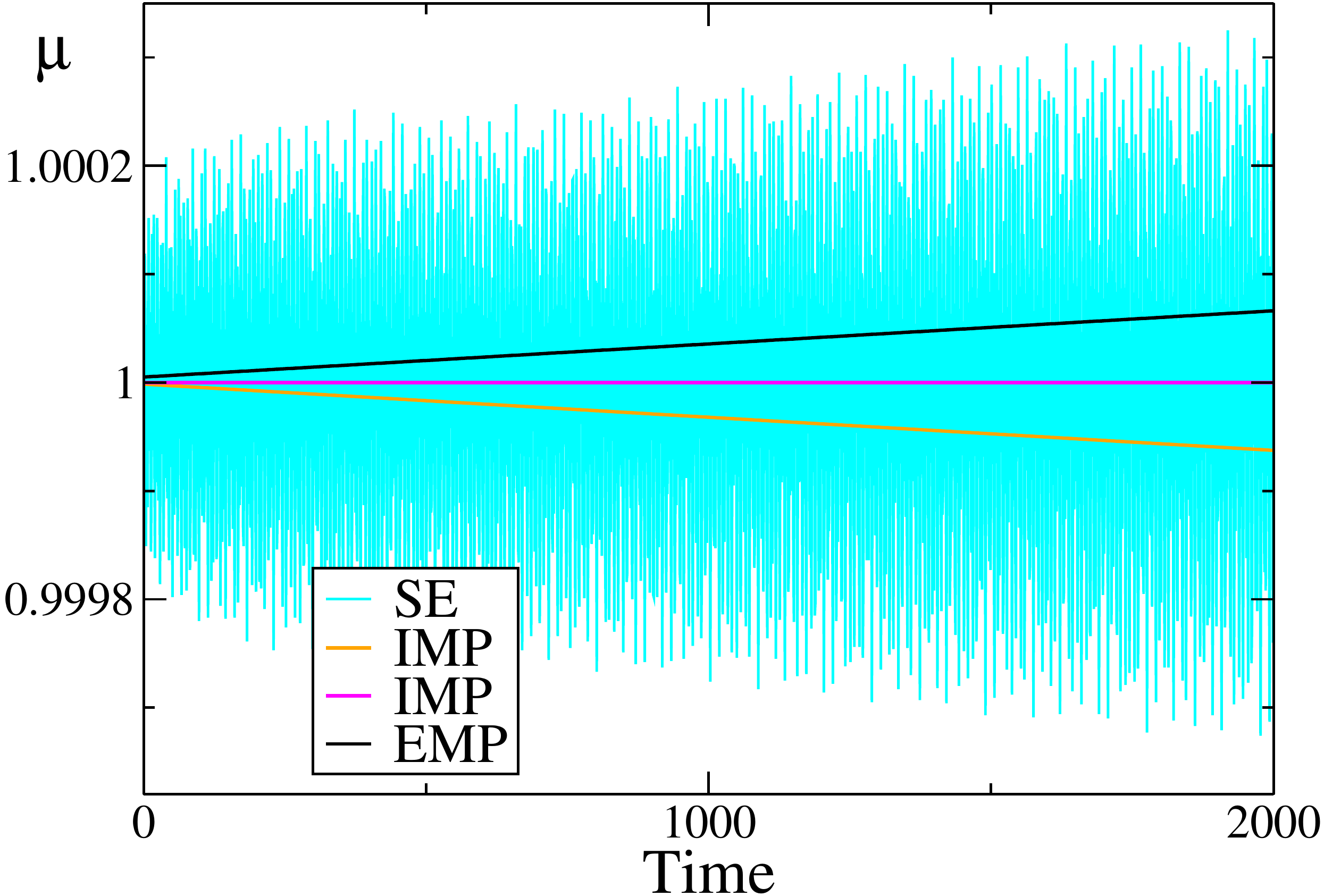}
  \caption{Conservation of Lagrange multiplier $\mu$ for different integration schemes: symplectic Euler (SE), implicit midpoint (IMP), explicit midpoint (EMP).
  The orange and magenta curves differ for the convergence threshold in the implicit solver ($10^{-6}$ and $10^{-12}$, respectively).
  It is possible to avoid energy drift only when applying IMP (with sufficiently small threshold value) and SE. For all the simulations $K=0.5$.}
  \label{fig.0}
\end{figure}

\begin{figure}[htb]
  \centering
  \includegraphics[width=0.5\textwidth]{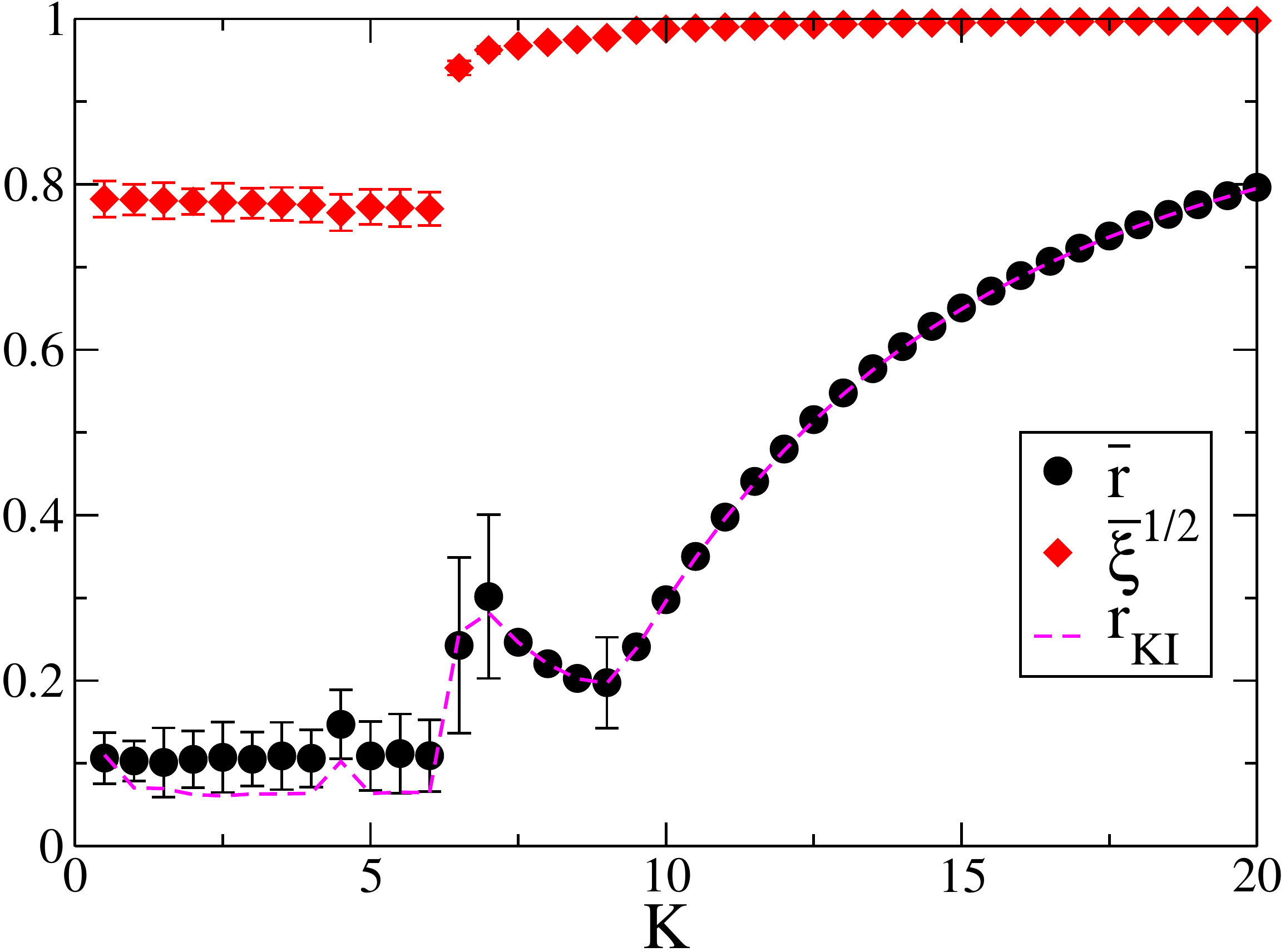}
  \caption{Average order parameters vs the coupling constant $K$, with perfectly balanced bimodal distribution of the $\Omega_j$.
  $r_{\text{KI}}$ is the average standard Kuramoto parameter, calculated through a $4$\textsuperscript{th} order Runge-Kutta scheme with the original formulation.
  $\bar{r}$ and ${\bar{\xi}}^{1/2}$ are the average standard Kuramoto parameter and generalized order parameter, respectively, obtained with the explicit midpoint rule applied to the pHDAE formulation.
  } \label{fig.2}
\end{figure}


In order to understand the transition from the non-synchronized state at low coupling constant $K$ to the synchronized state at high $K$ values, we performed sequences of simulations by adiabatically increasing the parameter $K$ with (the same) random initial conditions for $\{\theta_i \}$ and $\{\dot{\theta}_i\}$, for both \eqref{eq:stdkura0} and \eqref{eq:stdkura2}.
For each value of $K$ but the first one, the simulation is initialized by employing the last configuration of the previous simulation in the sequence.
We calculated the average order parameters $\bar{r}$, ${\bar{\xi}}^{1/2}$ which both show a non-monotonic behavior in $K$ (see Figure~\ref{fig.2}).
In particular, for small coupling values $\bar{r}\propto 1/\sqrt{N}$, we observe an abrupt jump for $K=6.5$. 
Subsequently, $\bar{r}$ decreases, reaching a minimum at $K=9$.
Then, for larger $K$, the order parameter $\bar{r}$ increases monotonically towards the fully synchronized regime.
There are no substantial differences between the Kuramoto order parameter $r$ calculated within the original and the pHDAE formulation.
The average order parameter ${\bar{\xi}}^{1/2}$, on the other hand, does not show such an
irregular behavior for small coupling values, but has an almost constant value until $K=6.5$, where the transition to synchronization takes place and, from that value on, it rapidly increases towards $1$.
Since $\xi$ takes into account the topology of the network, by introducing the connectivity matrix into the definition and summing over the connecting edges, the investigation of the synchronization level is more straightforward: its behavior is more stable for small coupling constants, where it describes better than $r$ the real level of synchronization present in the network, while it is nonetheless able to identify the transition to synchronization at $K=6.5$, as the classical Kuramoto order parameter.

The correctness of the critical coupling constant value at which the transition to synchronization takes place can be confirmed indirectly by the calculation of the maximal Lyapunov exponent $\lambda_M$, which
represents a measure of the stability of the system and is a good indicator for the emergence of problems in the network. In particular, $\lambda_M$ shows bigger fluctuations for values of $K$ just below the transition to synchronization, and it becomes zero for $K\geq 6.5$ (see Figure~\ref{fig.3}).

\begin{figure}[ht!]
  \centering
  \includegraphics*[angle=0,width=0.5\textwidth]{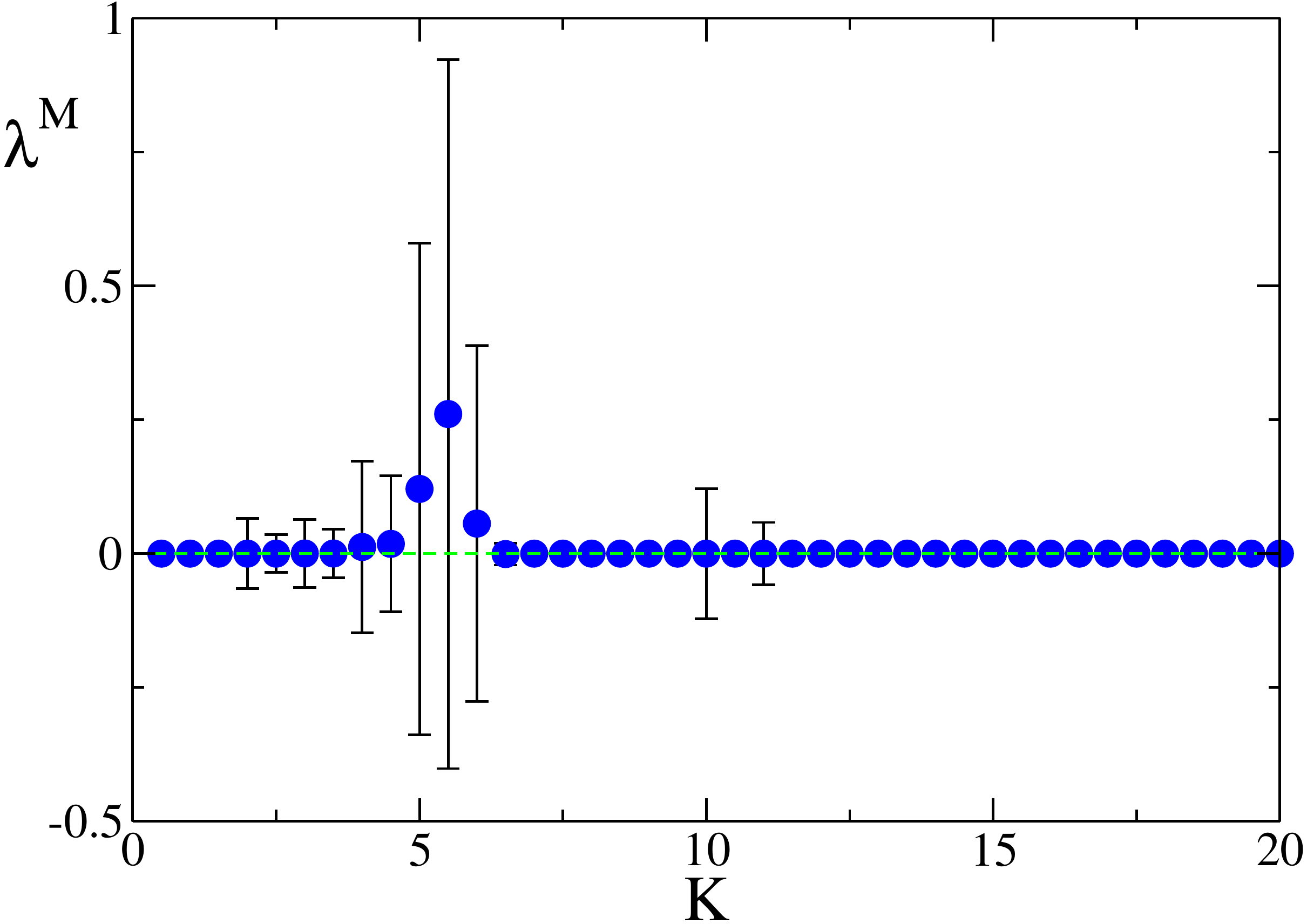}
  \caption{Maximal Lyapunov exponent vs coupling constant K, with perfectly balanced bimodal distribution of the $\Omega_j$. The dashed green line represents the value $\lambda^M=0$.
  } \label{fig.3}
\end{figure}

The reason why $\bar{r}$ fails in identifying the correct level of synchronization, can be understood by examining the average phase velocity of the oscillators $\langle\dot{\theta}_i\rangle$
(see Figure~\ref{fig.4}). For coupling constant $K<6.5$, the system is split in two clusters: one composed by the sources, which oscillates close to their proper frequency $G$, and another one containing the consumers, which rotates with negative average velocity.
Therefore, there is a non-trivial form of partial synchronization already present in the network that
is not reflected by $\bar{r}$, whose value of the order of $1/\sqrt{N}$ indicates that the system behaves asynchronously. 
For $K=6.5$, the coupling is sufficient to induce frequency adaptation and to enhance synchronization: the two clusters start merging to a single cluster, although a large part of the oscillators is still not synchronized. For $K=7$, the oscillators get more entrained, and most of them are locked with almost zero average velocity; however, a large part ($53$ out of $127$) form
a secondary cluster of whirling oscillators, with velocity $\langle\dot{\theta}\rangle\approx-0.122$. This secondary cluster has a geographical interpretation, since it includes the power stations and consumers located
in central and southern Italy, including Sicily, as already seen in \cite{OlmNBT14}.
By increasing the coupling to $K=7.5$, the two clusters merge in a unique cluster with few scattered oscillators and finally, for coupling $K\geq 8$, all the oscillators are locked in a unique cluster, which is reflected in a monotonic increase of the average order parameters.

\begin{figure}[ht]
  \centering
  \includegraphics[width=0.6\textwidth]{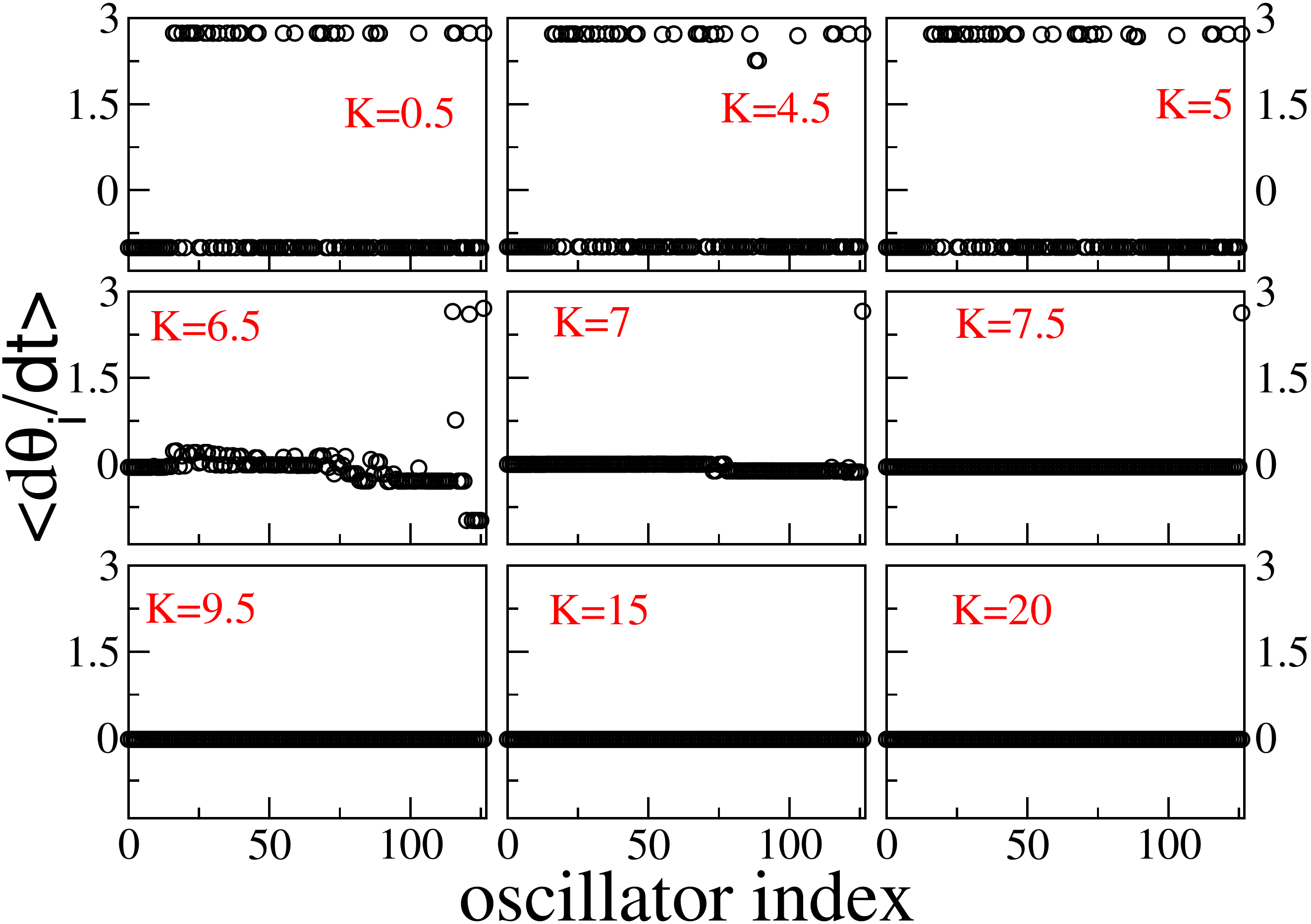}
  \caption{ Average phase velocity of each oscillator for different values of the coupling $K$, with perfectly balanced bimodal distribution of the $\Omega_j$. The data have been obtained by changing $K$ adiabatically from $K=0$, with $\Delta K=0.5$.
  } \label{fig.4}
\end{figure}

The investigation of the transition to synchronization has been done, up to this point, for fixed values of the mass $m$ and of the damping constant $d$.
In order to design control schemes capable of guiding the network towards synchronization, we investigate the response of the system for different masses and damping constants.
For bigger masses the critical coupling value at which the transition to synchronization takes place
becomes bigger, thus meaning that bigger coupling strength is necessary in order to achieve the synchronized state for power grid networks in case of bigger masses (see Figure~\ref{fig.5}).
\begin{figure}[ht!]
  \centering
  \includegraphics*[angle=0,width=0.6\textwidth]{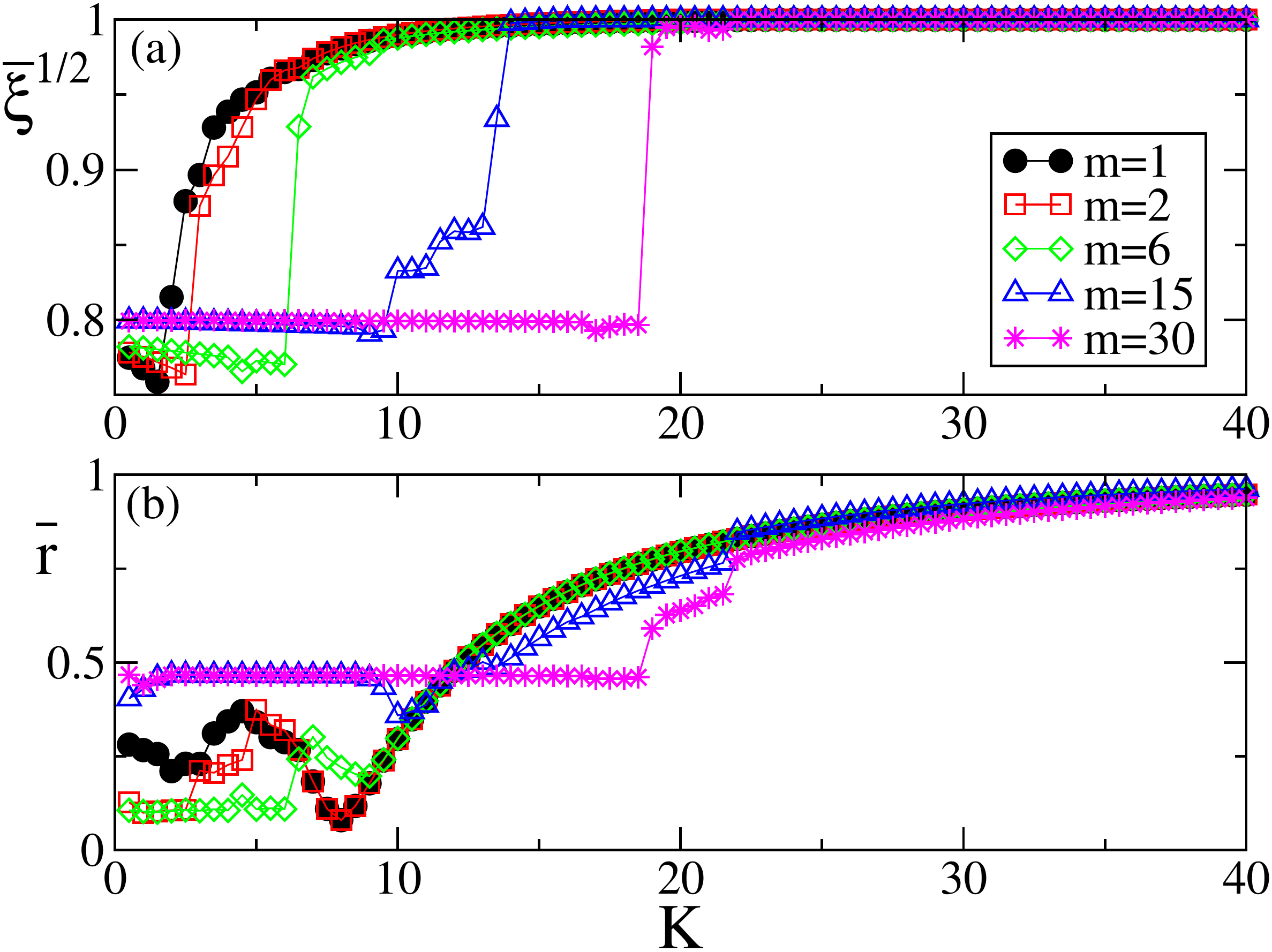}
  \caption{Average order parameters vs the coupling constant $K$, for different masses.
  The data have been obtained by changing $K$ adiabatically the coupling constant from $K=0$, with $\Delta K=0.5$. Increasing the mass shifts the critical value at which the system reaches synchronization.
  } \label{fig.5}
\end{figure} 
On the other hand, bigger damping constants help reaching synchronization, as shown in Figure~\ref{fig.6}:
for fixed coupling values, we can increase the synchronization level by increasing
the damping $d$. However, a consistent increase is possible only for coupling constants slightly smaller or bigger than the critical value $K=6.5$ at which we have observed the transition
in Figure~\ref{fig.2}, thus illustrating that the coupling strength plays a more important role in the transition, with respect to the damping constant.
This observation is also justified by the equivalent investigation of the dynamics of the system, obtained by varying the mass $m$, while keeping the coupling strength constant.
In particular, the order parameter ${\bar{\xi}}^{1/2}$ remains constant for fixed $K$, and does not show any dependence on the mass (results not shown).

\begin{figure}[ht]
  \centering
  \includegraphics[width=0.5\textwidth]{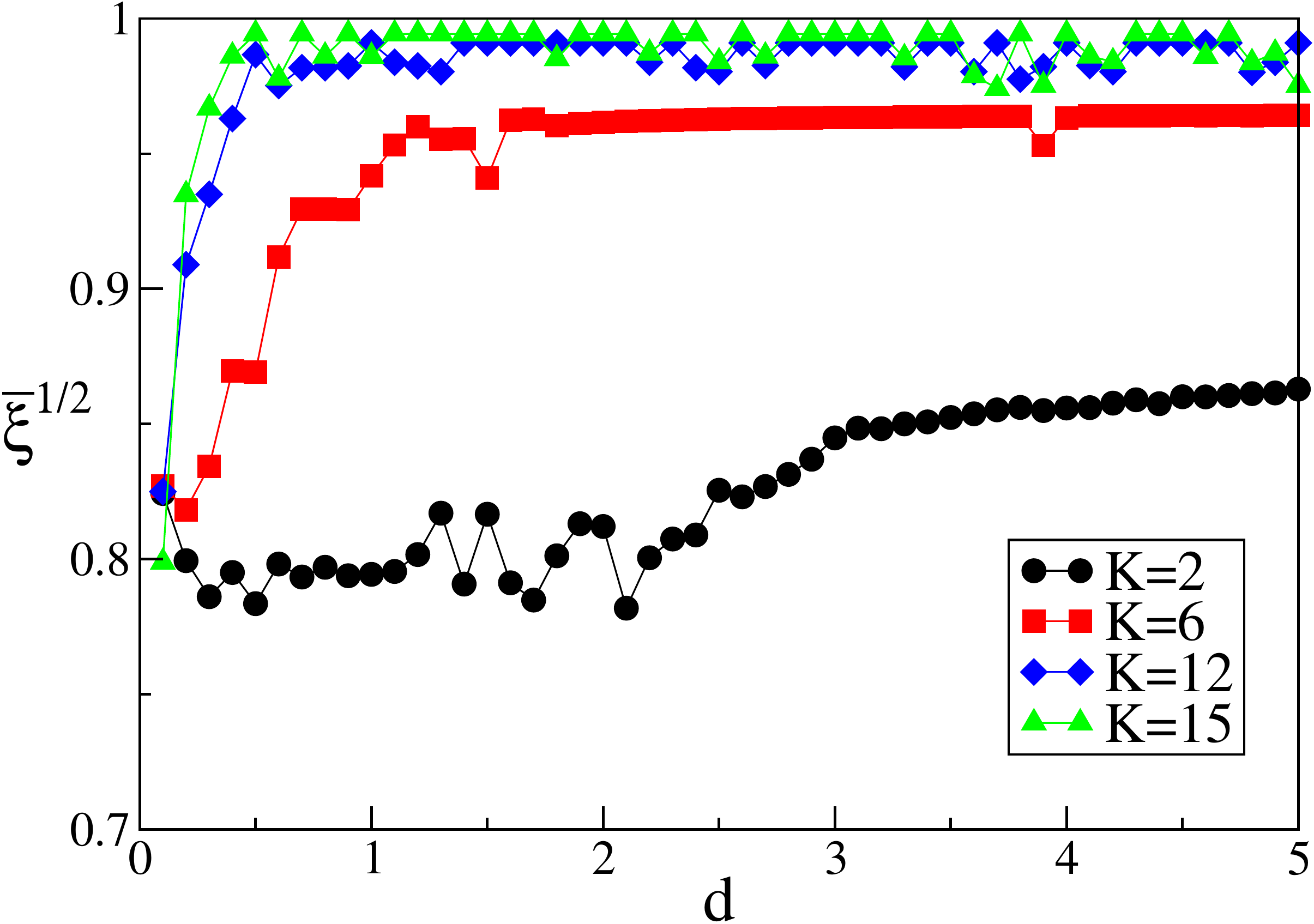}
  \caption{Average order parameter ${\bar{\xi}}^{1/2}$ vs the damping constant $d$ for different coupling constants $K$.
  The data have been obtained by changing adiabatically the coupling constant, starting from $K=0$ and with $\Delta K=0.5$.
  } \label{fig.6}
\end{figure}

When a perturbation is introduced into a power grid and a generator is diverging from the synchronized regime at which it is supposed to work, it is important to react to this perturbation as fast as possible in order to avoid shutting down the generator. The pHDAE formulation of the Kuramoto model with inertia is extremely useful for this kind of problem, thanks to the possibility of fast computation associated to a hierarchical model of the power grid: the differential equations \eqref{eq:stdkura2} guarantee an on-line fast description of the real system, irrespectively of the fact that it is a simplified model. In particular, the calculation of the maximal Lyapunov exponent can be used to identify instabilities emerging in the network. When it is necessary to operate at the level of the real grid, due to emerging disturbances that may affect the stability, the hierarchical model then allows to control and adjust the parameters in more sophisticated models, like the instantaneous power model
shown in Section~\ref{sec:Modelhiearchy}, and finally in the real grid, by advancing upward along the hierarchy.
Starting from a situation where the system is slightly out of synchronization (that is modelled by
choosing $K=6$), we can thus design a control method once the response of the system to the parameters change is known.
To illustrate this, we have calculated the time $t_s$ that the system needs to reach
the synchronization from this out-of-synchronicity initial condition: if the system is not operating perfectly, we can restore the desired status as fast as possible by abruptly increasing the coupling
constant or the damping constant or by decreasing the mass (see Figure~\ref{fig.7}).

\begin{figure}[ht!]
  \centering
  \includegraphics*[angle=0,width=0.6\textwidth]{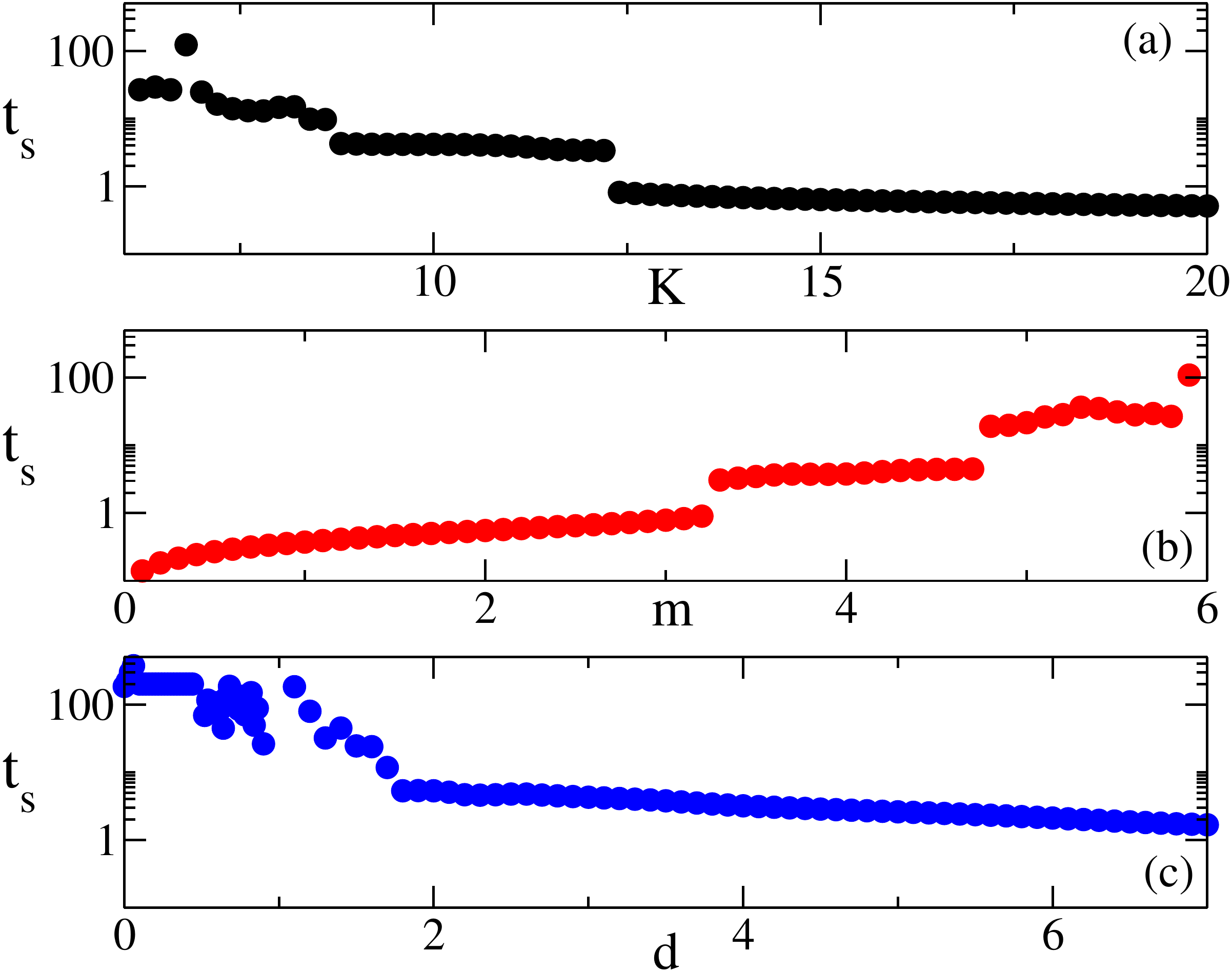}
  \caption{Time $t_s$ necessary to reach a good synchronization state ($\xi^{1/2}>0.95$), as a function of the parameters. The values $K$, $m$, $d$ reported on the $x$-axes of the panels represent the values to which the initial parameters are tuned, in order to calculate the time needed to reach synchronization.
  } \label{fig.7}
\end{figure}

\subsection{Gaussian bimodal distribution}

Since in a real power grid we do not expect the power of the generators (or consumers) to be characterized by the same exact values, we analyzed also the case where the $\Omega_j$ are chosen as i.i.d.~random variables, whose distribution is given as the combination of two almost non-overlapping Gaussian distributions
\[
  g(s) = \frac{1}{2\sqrt{2\pi}} \left[e^{\frac{-(s-\Omega^-)^2}{2}} + e^{\frac{-(s-\Omega^+)^2}{2}}\right],
\]
centered at values $\Omega^-,\Omega^+\in\mathbb R$, in our case $\Omega^\pm=\pm2$.
The other parameters are kept as before.
The analysis that we have done is equivalent to what is presented in Figures~\ref{fig.2}, \ref{fig.3} and \ref{fig.4}: we investigated the transition to synchronization and the characterization of the different dynamical behavior, emerging for different coupling constants.

The calculation of the average order parameters $\bar{r}$, ${\bar{\xi}}^{1/2}$ as a function of the coupling constant reveals that, in this setup, it is more difficult to achieve synchronization,
due to the inhomogeneity of the natural frequencies (see Figure~\ref{fig.8}). The classical order parameter $\bar{r}$ is still irregular and unstable, irrespectively whether we use the original Kuramoto model or the pHDAE formulation,
while ${\bar{\xi}}^{1/2}$ is more stable and informative. If we concentrate on the behavior of ${\bar{\xi}}^{1/2}$, we observe a continuous transition to synchronization, instead than a jump from partial to full synchronization.
However, in comparison with the setup of Figure~\ref{fig.2}, a bigger coupling constant is needed in order to obtain the same level of synchronization.

\begin{figure}[ht!]
  \centering
  \includegraphics[width=0.5\textwidth]{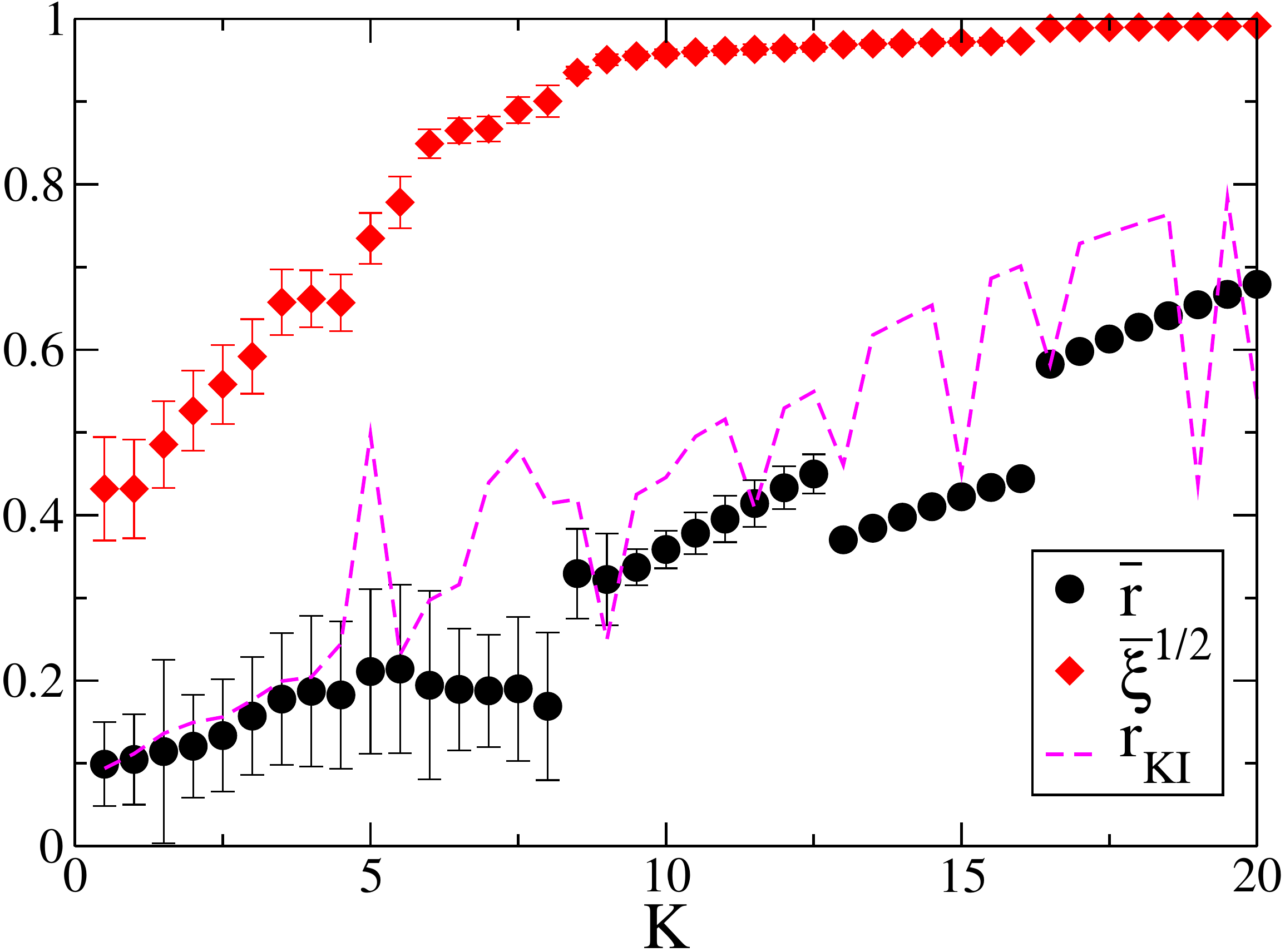}
  \caption{Average order parameters vs the coupling constant $K$, with Gaussian distributions of the $\Omega_j$.
  The integration schemes used are the same as in Figure~\ref{fig.2}.
  } \label{fig.8}
\end{figure}

The investigation of the transition in terms of the maximal Lyapunov exponent gives more insight into the transition point from chaotic behavior to synchronization. The critical value of the
coupling constant at which the system synchronizes is $K_c=8.75 \pm 0.25$: for $K>K_c$ the system is stable, as testified by the value $\lambda^M=0$, while for $K<K_c$ the system is chaotic
for a wide range of the coupling constant (much wider than in Figure~\ref{fig.3}).

\begin{figure}[ht!]
  \centering
  \includegraphics[width=0.5\textwidth]{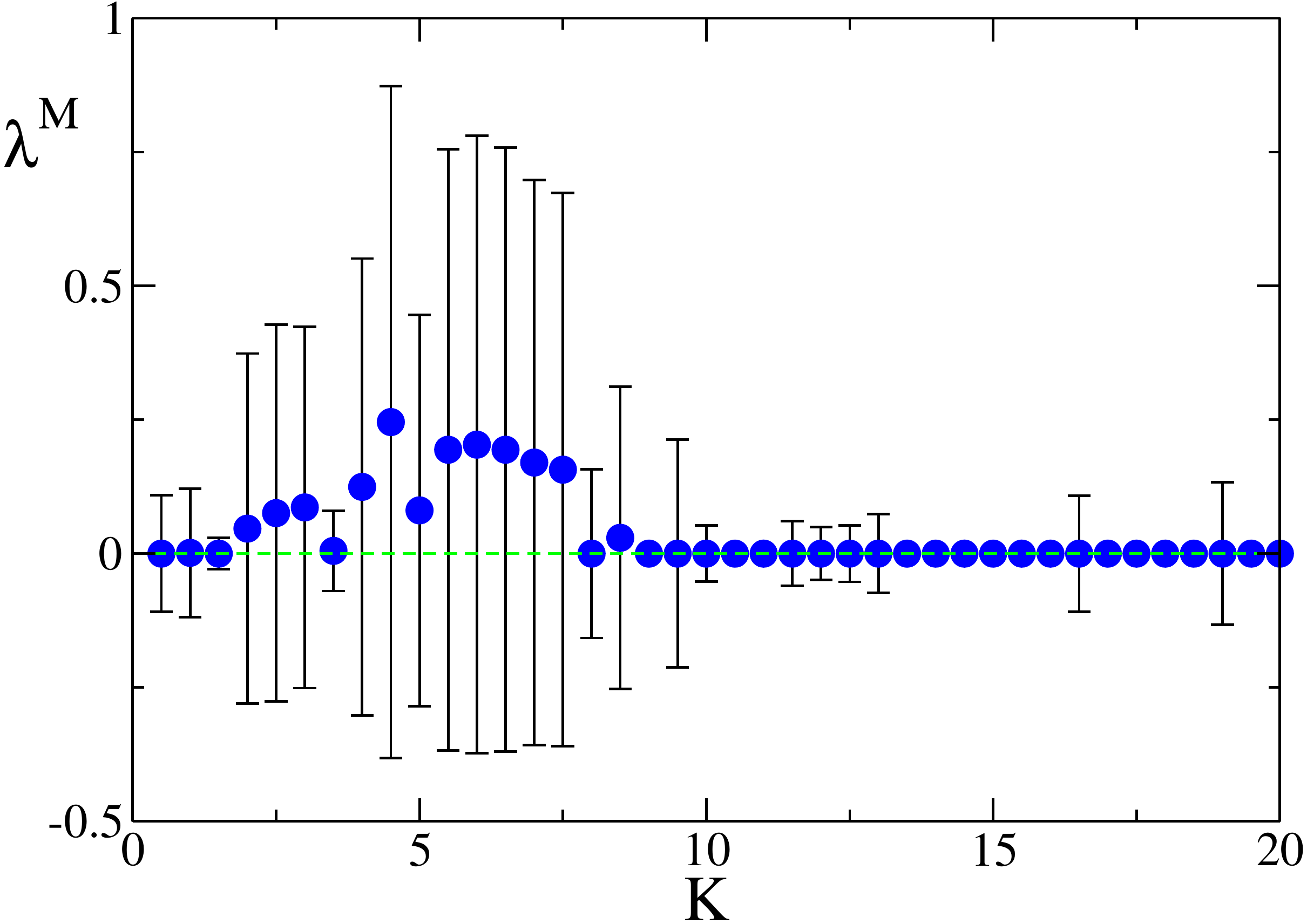}
  \caption{Maximal Lyapunov exponent vs the coupling constant $K$, with Gaussian distributions of the $\Omega_j$.
  The dashed green line represents the value $\lambda^M=0$.
  } \label{fig.9}
\end{figure}

Finally, the average phase velocities $\langle\dot{\theta}_i\rangle$ of all oscillators are reported monotonically in Figure~\ref{fig.10}.
For small values of the coupling constant $K$, the velocities are erratic and each oscillator moves independently from the others.
Starting around $K=3.5$, some clusters emerge in the network: by increasing the coupling constant, some oscillators tend to merge in larger and larger clusters, while others (especially the ones with higher absolute velocity) remain asynchronous.
By $K\geq 8.5$, most oscillators have assembled in few, big clusters.
In particular, for $K=8.5$ we have a $3$-cluster state plus a limited set of asynchronous oscillators, thus meaning that there are at least three effective degrees of freedom acting into the system and contributing to the dynamics.
With three degrees of freedom, it is possible to observe a chaotic motion (see \cite{WatS94}), as confirmed by the positive value of $\lambda^M$ for this coupling value.
On the other hand, by $K=9.5$ the system has merged into a $2$-cluster state, thus it is not possible to observe chaos anymore.
In particular, the dimension of these clusters is very asymmetric, and one cluster is much bigger then the other one, thus justifying the increased level of synchronization and the corresponding transition to the synchronized state.
By $K=13$, all oscillators have collapsed into a unique cluster, and we expect that to be the case for even larger coupling constants.

\begin{figure}[ht!]
  \begin{center}
  \includegraphics[width=0.8\textwidth]{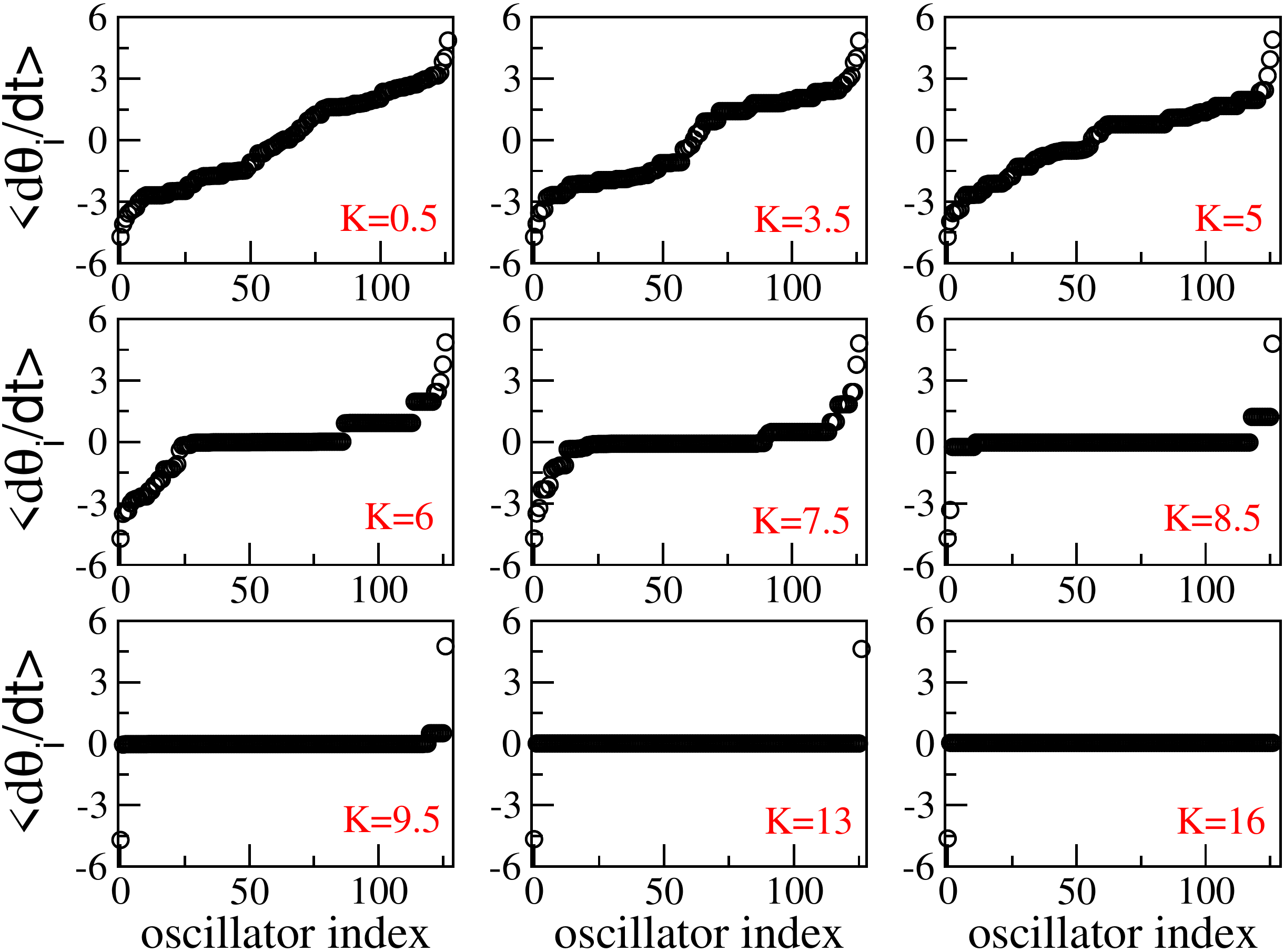}
  \end{center}
  \caption{Average phase velocity of each oscillator $\langle\dot{\theta}_i\rangle$ vs the oscillator index, for different values of the coupling K, with Gaussian distributions of the $\Omega_j$. The indices of the oscillators have been rearranged in such a way that the average phase velocities are monotonically increasing. The data have been obtained by changing $K$ adiabatically from $K=0$, with $\Delta K=0.5$.
  } \label{fig.10}
\end{figure}

\section{Summary and Outlook}
We have presented a new port-Hamiltonian differential-algebraic formulation of the Kuramoto model of coupled oscillators as well as a new definition of the order parameter. 
The new model has several advantages, it is easily extended to models of finer granularity as they are used in qualitative stability and synchrony analysis of power systems. 
The new order parameter is more robust in limiting situations. We have also illustrated the advantage of the port-Hamiltonian formulation in the preservation of conserved quantities. 
The new approach and its advantages have been illustrated with many numerical examples carried out for a semi-realistic model of the Italian power grid.

Future work will include the analysis of the whole model hierarchy, error control in adaptive time step and model selection, as well as model reduction techniques that allow real 
time stabilization and synchronization as well as the incorporation of switching in and time-delay in the model.

\section{Acknowledgement}
V.~M.~acknowledges the Deutsche Forschungsgemeinschaft via Project A2 within SFB 910 and the {\it Einstein Foundation Berlin} via the Einstein Center ECMath.
R.~M.~is supported by {\it Einstein Foundation Berlin} via the Einstein Center ECMath. 
S.~O.~and E.~S.~were supported by DFG via Project A1 in the framework of SFB 910.

\bibliographystyle{plain}
\bibliography{MehMOS}

\begin{thebibliography}{10}

\bibitem{BeaG11}
C.~Beattie and S.~Gugercin.
\newblock Structure-preserving model reduction for nonlinear port-{H}amiltonian
  systems.
\newblock In {\em 50th IEEE Conference on Decision and Control and European
  Control Conference (CDC-ECC), 2011}, pages 6564--6569. IEEE, 2011.

\bibitem{BeaMXZ17}
C.~Beattie, V.~Mehrmann, H.~Xu, and H.~Zwart.
\newblock Port-hamiltonian descriptor systems.
\newblock Preprint 06-2017, Institut f{\"u}r Mathematik, TU Berlin, 2017.

\bibitem{Bre08}
P.~C. Breedveld.
\newblock {\em Modeling and Simulation of Dynamic Systems using Bond Graphs},
  pages 128--173.
\newblock EOLSS Publishers Co. Ltd./UNESCO, Oxford, UK, 2008.

\bibitem{ByrIW91}
C.~I. Byrnes, A.~Isidori, and J.~C. Willems.
\newblock Passivity, feedback equivalence, and the global stabilization of
  minimum phase nonlinear systems.
\newblock {\em IEEE Trans. Autom. Control}, 36:1228--1240, 1991.

\bibitem{CamKM11}
S.~L. Campbell, P.~Kunkel, and V.~Mehrmann.
\newblock Regularization of linear and nonlinear descriptor systems.
\newblock In L.~T. Biegler, S.~L. Campbell, and V.~Mehrmann, editors, {\em
  Control and Optimization with Differential-Algebraic Constraints}, pages
  17--36. 2012.

\bibitem{CerSB07}
J.~Cervera, A.J. van~der Schaft, and A.~Ba{\~n}os.
\newblock Interconnection of port-{H}amiltonian systems and composition of
  {D}irac structures.
\newblock {\em Automatica}, 43:212--225, 2007.

\bibitem{CouJMTB08}
F.~Couenne, C.~Jallut, B.~M. Maschke, M.~Tayakout, and P.~C. Breedveld.
\newblock Bond graph for dynamic modelling in chemical engineering.
\newblock {\em Chemical engineering and processing, Elsevier, Amsterdam},
  47:1994--2003, 2008.

\bibitem{DorCB13}
F.~Doerfler, M.~Chertkov, and F.~Bullo.
\newblock Synchronization in complex oscillator networks and smart grids.
\newblock {\em Proceedings of the National Academy of Sciences.},
  6(110):2005--10, 2013.

\bibitem{FilNP08}
G.~Filatrella, A.~H. Nielsen, and N.~F. Pedersen.
\newblock Analysis of a power grid using a {K}uramoto-like model.
\newblock {\em The European Physical Journal B}, 61(4):485--491, 2008.

\bibitem{ForFSF12}
L.~Fortuna, M.~Frasca, and A.~Sarra Fiore.
\newblock A network of oscillators emulating the italian high-voltage power
  grid.
\newblock {\em Int. J. Modern Phys.}, 26:25, 2012.

\bibitem{GolSBM03}
G.~Golo, A.~J.~{van der} Schaft, P.~C. Breedveld, and B.~M. Maschke.
\newblock {H}amiltonian formulation of bond graphs.
\newblock In A.~Rantzer {R. Johansson}, editor, {\em Nonlinear and Hybrid
  Systems in Automotive Control}, pages 351{\~n}--372. Springer, Heidelberg,
  2003.

\bibitem{GolV96}
G.~H. Golub and C.~F. {Van~Loan}.
\newblock {\em Matrix Computations}.
\newblock Johns Hopkins Univ. Press, Baltimore, 3rd edition, 1996.

\bibitem{GugPBS12}
S.~Gugercin, R.~V. Polyuga, C.~Beattie, and A.~J.~{van der} Schaft.
\newblock Structure-preserving tangential interpolation for model reduction of
  port-{H}amiltonian systems.
\newblock {\em Automatica}, 48:1963--1974, 2012.

\bibitem{HinP05}
D.~Hinrichsen and A.~J. Pritchard.
\newblock {\em Mathematical Systems Theory {I}. Modelling, State Space
  Analysis, Stability and Robustness}.
\newblock Springer-Verlag, New York, NY, 2005.

\bibitem{JacZ12}
B.~Jacob and H.~Zwart.
\newblock {\em Linear port-{H}amiltonian systems on infinite-dimensional
  spaces}.
\newblock Operator Theory: Advances and Applications, 223.
  Birkh{\"a}user/Springer Basel AG, Basel CH, 2012.

\bibitem{KunBL94}
P.~Kundur, N.J. Balu, and M.G. Lauby.
\newblock {\em Power system stability and control}.
\newblock EPRI power system engineering series. McGraw-Hill, 1994.

\bibitem{KunM06}
P.~Kunkel and V.~Mehrmann.
\newblock {\em Differential-Algebraic Equations. Analysis and Numerical
  Solution}.
\newblock EMS Publishing House, Z{\"u}rich, Switzerland, 2006.

\bibitem{KunM07}
P.~{Kunkel} and V.~{Mehrmann}.
\newblock Stability properties of differential-algebraic equations and
  spin-stabilized discretizations.
\newblock {\em Electron. Trans. Numer. Anal.}, 26:385--420, 2007.

\bibitem{LamMT13}
R.~Lamour, R.~M{\"a}rz, and C.~Tischendorf.
\newblock {\em Differential-algebraic equations: a projector based analysis}.
\newblock Springer Science \& Business Media, 2013.

\bibitem{NisM15}
T.~Nishikawa and A.~E. Motter.
\newblock Comparative analysis of existing models for power-grid
  synchronization.
\newblock {\em New Journal of Physics}, 1(17):015012, 2015.

\bibitem{OlmNBT14}
S.~Olmi, A.~Navas, S.~Boccaletti, and A.~Torcini.
\newblock Hysteretic transitions in the {K}uramoto model with inertia.
\newblock {\em Phys. Rev. E}, 90:042905, Oct 2014.

\bibitem{OlmT16}
S.~Olmi and A.~Torcini.
\newblock {\em Dynamics of Fully Coupled Rotators with Unimodal and Bimodal
  Frequency Distribution}, pages 25--45.
\newblock Springer International Publishing, Cham, 2016.

\bibitem{OrtSMM01}
R.~Ortega, A.~J.~{van der} Schaft, Y.~Mareels, and B.~M. Maschke.
\newblock Putting energy back in control.
\newblock {\em Control Syst. Mag.}, 21:18--{\~n}33, 2001.

\bibitem{PolS10}
R.~V. Polyuga and A.~J.~{van der} Schaft.
\newblock Structure preserving model reduction of port-{H}amiltonian systems by
  moment matching at infinity.
\newblock {\em Automatica}, 46:665--672, 2010.

\bibitem{RohSTW12}
M.~Rohden, A.~Sorge, M.~Timme, and D.~Witthaut.
\newblock Self-organized synchronization in decentralized power grids.
\newblock {\em Physical Review Letters}, 6(109):064101, 2012.

\bibitem{SalMV84}
F.~Salam, J.~Marsden, and P.~Varaiya.
\newblock Arnold diffusion in the swing equations of a power system.
\newblock {\em IEEE Trans. Circ. Syst.}, 8(31):673--688, 1984.

\bibitem{Sch04}
A.~J.~{van der} Schaft.
\newblock Port-{H}amiltonian systems: network modeling and control of nonlinear
  physical systems.
\newblock In {\em Advanced Dynamics and Control of Structures and Machines},
  {CISM} Courses and Lectures, Vol. 444. Springer Verlag, New York, N.Y., 2004.

\bibitem{Sch06}
A.~J.~{van der} Schaft.
\newblock Port-{H}amiltonian systems: an introductory survey.
\newblock In J.~L.~Verona {M. Sanz-Sole} and J.~Verdura, editors, {\em Proc. of
  the International Congress of Mathematicians, vol. III, Invited Lectures},
  pages 1339--1365, Madrid, Spain, 2006.

\bibitem{Sch13}
A.~J.~{van der} Schaft.
\newblock Port-{H}amiltonian differential-algebraic systems.
\newblock In {\em Surveys in Differential-Algebraic Equations I}, pages
  173--226. Springer-Verlag, 2013.

\bibitem{Var11}
R.S. Varga.
\newblock {\em Geršgorin and His Circles}.
\newblock Springer Series in Computational Mathematics. Springer Berlin
  Heidelberg, 2004.

\bibitem{WatS94}
S.~Watanabe and S.~H. Strogatz.
\newblock Constants of motion for superconducting josephson arrays.
\newblock {\em Physica D: Nonlinear Phenomena}, 3-4(74):197--253, 1994.

\bibitem{Wil72a}
J.~C. Willems.
\newblock Dissipative dynamical systems -- {P}art {I}: {G}eneral theory.
\newblock {\em Arch. Ration. Mech. Anal.}, 45:321--351, 1972.

\bibitem{Wil72b}
J.~C. Willems.
\newblock Dissipative dynamical systems -- {P}art {II}: {L}inear systems with
  quadratic supply rates.
\newblock {\em Arch. Ration. Mech. Anal.}, 45:352--393, 1972.

\end{thebibliography}

\end{document}